\documentclass[11pt,reqno]{amsart}

\usepackage{amsfonts,amsmath,amssymb,wasysym}
\usepackage[margin=1in]{geometry}

\usepackage{algorithm}
\usepackage[noend]{algpseudocode}

\usepackage{graphicx}
\newtheorem{theorem}{Theorem}[section]
\newtheorem{proposition}[theorem]{Proposition}
\newtheorem{lemma}[theorem]{Lemma}
\newtheorem{corollary}[theorem]{Corollary}
\newtheorem{definition}[theorem]{Definition}
\newtheorem{observation}[theorem]{Observation}

\newtheorem*{rep@theorem}{\rep@title}
\newcommand{\newreptheorem}[2]{%
\newenvironment{rep#1}[1]{ \def\rep@title{#2 \ref{##1}} \begin{rep@theorem}} {\end{rep@theorem}} }

\newreptheorem{theorem}{Theorem}
\newcommand{\eps}{\varepsilon}
\newcommand{\qmin}{q_{\min}}

\newcommand{\etamax}{\eta}
\newcommand{\pmax}{p_{\max}}
\newcommand{\pmin}{p_{\min}}

\newcommand{\bE}{\ensuremath{\mathbf{E}}}

\DeclareMathOperator{\poly}{poly}
\DeclareMathOperator{\polylog}{polylog}
\DeclareMathOperator{\maxsize}{maxsize}
\DeclareMathOperator{\depth}{depth}
\DeclareMathOperator{\sink}{sink}
\DeclareMathOperator{\var}{Var}

\DeclareMathOperator{\size}{size}
\DeclareMathOperator{\concat}{Concat}
\begin{document}

\begin{abstract}
  The Lov\'{a}sz Local Lemma (LLL) is a keystone principle in probability theory, guaranteeing the existence of configurations which avoid a collection $\mathcal B$ of ``bad'' events which are mostly independent and have low probability. In its simplest ``symmetric'' form, it asserts that whenever a bad-event has probability $p$ and affects at most $d$ bad-events, and $e p d < 1$, then a configuration avoiding all $\mathcal B$ exists.

  A seminal algorithm of Moser \& Tardos (2010) (which we call the MT algorithm) gives nearly-automatic randomized algorithms for most constructions based on the LLL. However, deterministic algorithms have lagged behind. We address three specific shortcomings of the prior deterministic algorithms.   First, our algorithm applies to the LLL criterion of Shearer (1985); this is more  powerful than alternate LLL criteria and also removes a number of nuisance parameters and leads to cleaner and more legible bounds.   Second, we provide parallel algorithms with much greater flexibility in the functional form of the bad-events. Third, we provide a derandomized version of the MT-distribution, that is, the distribution of the variables at the termination of the MT algorithm.

We show applications to non-repetitive vertex coloring, independent transversals, strong coloring, and other problems. These give deterministic algorithms which essentially match the best previous randomized sequential and parallel algorithms. 
\end{abstract}
\author[David G. Harris]{David G. Harris$^1$}
\setcounter{footnote}{0}
\addtocounter{footnote}{1}
\footnotetext{Department of Computer Science, University of Maryland, 
College Park, MD 20742. 
Email: \texttt{davidgharris29@gmail.com}.}

\title[Deterministic algorithms for the LLL]{Deterministic algorithms for the Lov\'{a}sz Local Lemma: simpler, more general, and more parallel}

\maketitle

This is an extended version of a paper which appeared in the ACM-SIAM Symposium on Discrete Algorithms (SODA) 2022.

\section{Introduction}
\label{r1fmt-sec}
The Lov\'{a}sz Local Lemma (LLL) is a keystone principle in probability theory. It asserts that, in a probability space $\Omega$ provided with a set $\mathcal B$ of ``bad-events'', which are not too likely and are not too interdependent (in a certain technical sense), there is a positive probability that no event in $\mathcal B$ occurs.  The simplest ``symmetric'' form of the LLL states that if each bad-event $B \in \mathcal B$  affects at most $d$ bad-events and has probability at most $\frac{1}{e d}$,  then $\Pr( \bigcap_{B \in \mathcal B} \overline{B} ) > 0$. In particular, a configuration avoiding $\mathcal B$ exists. Other forms of the LLL, such as the ``asymmetric'' criterion, can be used when the bad-events have a more complex dependency structure.

Although the LLL applies to general probability spaces, a simpler ``variable version'' covers most applications to combinatorics and graph theory: the probability space $\Omega$ is defined by $n$ independent variables $X(1), \dots, X(n)$ over an alphabet $\Sigma$, and each bad-event $B \in \mathcal B$ is a boolean function $f_B$ of some variable subset $\var(B) \subseteq \{1, \dots, n\}$. In this setting, bad-events $B, B'$ are dependent in the sense of the LLL, which we denote by  $B \sim B'$, if $\var(B) \cap \var(B') \neq \emptyset$.

For any event $B$, we write $p(B) = \Pr_{\Omega}(B)$ and we define the \emph{neighborhood} $\Gamma(B)$ to be the set of bad-events $B' \in \mathcal B$ with $B \sim B'$ and $B' \neq B$.  We also define $\overline \Gamma(B)$ to be the inclusive neighborhood, i.e. $\overline \Gamma(B) = \Gamma(B) \cup \{B \}$. Finally, we define parameters
$$
d = \max_{B \in \mathcal B} |\overline \Gamma(B)|, \qquad \pmin = \min_{B \in \mathcal B} p(B),\qquad \pmax = \max_{B \in \mathcal B} p(B)
$$

We say a configuration $X = X(1), \dots, X(n)$ is \emph{good} if it avoids all the bad-events, i.e. if $f_B(X) = 0$ for all $B \in \mathcal B$. We let $m = | \mathcal B |$ and $\sigma = |\Sigma|$. 

As a paradigmatic example of the variable-setting LLL, consider a SAT instance where each clause contains $k$ literals and shares variables with at most $L$ other clauses. Here, the probability space $\Omega$ sets each variable to true or false independently with probability $1/2$, and has a bad-event corresponding to each clause being violated. These bad-events have probability $p = 2^{-k}$ and dependency $d = L + 1$. Thus, as long as $L \leq 2^k/e - 1$, a satisfying assignment exists.

The LLL, in its classical probabilistic form, only shows an exponentially small probability that a configuration is good; thus, it does not give efficient algorithms. In a seminal paper \cite{mt}, Moser \& Tardos introduced a simple randomized process, which we refer to as the \emph{MT algorithm}, to give efficient algorithms for nearly all LLL applications.

\begin{algorithm}[H]
\centering
\begin{algorithmic}[1]
\State Draw $X$ from the distribution $\Omega$
\While{some bad-event is true on $X$}
\State Arbitrarily select some true bad-event $B$
\State For each $i \in \var(B)$, draw $X(i)$ from its distribution in $\Omega$. 
\EndWhile
\end{algorithmic}
\caption{The MT algorithm}
\label{mt-seq}
\end{algorithm}

Under nearly the same conditions as the probabilistic LLL, the MT algorithm terminates in polynomial expected time. Moser \& Tardos also described a parallel version of this algorithm.

\subsection{Derandomized LLL algorithms}
For deterministic algorithms, the situation is not as clean.   The original paper of Moser \& Tardos described a deterministic variant for $d \leq O(1)$. This was strengthened by Chandrasekaran, Goyal \& Haeupler \cite{mt2} to unbounded $d$ under a stronger LLL criterion $e \pmax d^{1+\eps} \leq 1$, for any constant $\eps > 0$. For example, the latter criterion applies to $k$-SAT instances when $L \leq 2^{(1-\eps) k}/e - 1$. We note that, for many applications of the LLL, the bad-events are determined by sums of independent random variables. Due to the exponential decay appearing in, e.g. Chernoff bounds, it is relatively straightforward to ensure that such events can also satisfy the stronger LLL criterion.

Many of the deterministic algorithms can be parallelized. For the symmetric LLL, the algorithm of \cite{mt2} concretely has complexity of $O(\log^3(mn))$ time and $\poly(m,n)$ processors on an EREW PRAM. (We refer to this type of parallel deterministic complexity as $NC^3$.)  Alternate LLL derandomization algorithms, with a slightly faster run-time of $O(\log^2(mn))$, were described in \cite{mt3, polylog-paper}. As we will discuss shortly, these parallel deterministic algorithms have many additional technical conditions  on the types of bad-events that are allowed.

\subsection{Our contribution and overview}
We present new deterministic algorithms for the LLL.  There are three main contributions: (1) simpler and more general convergence criteria; (2) a parallel algorithm allowing more flexibility in the bad-events; (3) derandomization of a probabilistic method known as the \emph{MT-distribution}. Let us summarize how these improve over previous algorithms.

\bigskip

\noindent \textbf{Simpler criteria.} Previous MT derandomization algorithms have mostly focused on criteria analogous to the symmetric LLL, for example, the criterion $e \pmax d^{1+\eps} \leq 1$. Some algorithms also cover the asymmetric LLL, but in this case, there are many parameters that need to be checked. These parameters are related to technical conditions for the LLL without clear probabilistic interpretations. For example, the algorithm of \cite{mt2} requires, among other conditions, that the function $x: \mathcal B \rightarrow (0,1)$ witnessing the asymmetric LLL is bounded (both from above and below) by polynomials. While such conditions are usually satisfied in applications, they can be daunting  for non-specialists.

Our new algorithm modifies the algorithm of \cite{mt2} by tracking convergence in terms of a bound  known as \emph{Shearer's criterion}  \cite{shearer}. To summarize briefly, we first enumerate a relatively small collection of low-probability events which might cause the MT algorithm to fail to converge. These are based on ``witness dags'', which are a slight variant of the ``witness trees'' considered in \cite{mt2}. We then use conditional expectations to find a resampling table which causes all such events to be false. At this point, the MT algorithm can be simulated directly.

From a conceptual point of view, this is a relatively minor change to the algorithm. Mostly, it requires more careful counting of witness dags. However, it has two major advantages. First, Shearer's criterion is essentially the most powerful LLL-type criterion in terms of the probabilities and dependency structure of the bad-events. As a result, our deterministic algorithm covers nearly all applications of the probabilistic form of the LLL. In particular, it subsumes the symmetric and asymmetric criteria as well as other criteria such as the cluster-expansion criterion \cite{bissacot}.

But, there is a more important advantage: Shearer's criterion is more technically ``robust'' than the asymmetric LLL. The latter has a number of problematic corner cases, involving degenerate settings of certain variables. This is one of the main reasons the bounds in \cite{mt2} are so complex. We can obtain much more legible bounds which are easier to check and apply. For example,  we get the following crisp result:
\begin{theorem}
  \label{thm1x}
If the vector of probabilities $p(B)^{1-\eps}$ satisfies Shearer's criterion for constant  $\eps > 0$, and $\pmin \geq 1/\poly(n)$, then we can find a good configuration in polynomial time.
\end{theorem}
\begin{corollary}
  \label{cor1x}
  If $e \pmax d^{1+\eps} \leq 1$ then we can find a good configuration in polynomial time.
\end{corollary}
Corollary~\ref{cor1x} matches the result of \cite{mt2}. Theorem~\ref{thm1x} is much more general, as well as significantly simpler, compared to the asymmetric LLL criteria in \cite{mt2}. We describe a number of other simple LLL criteria, which we hope should be much easier to check for applications.

\bigskip

\noindent \textbf{A more general parallel algorithm.}
Our second main contribution is a new $NC^2$ parallel algorithm. This in turn gives parallel algorithms for many combinatorial constructions. To explain the novelty here, let us discuss an important limitation of the existing deterministic parallel algorithms: the indicator functions $f_B$ for the bad-events must satisfy additional structural properties. Specifically, the algorithm of \cite{mt2} requires the bad-events to be computable via decision-trees of depth $O(\log d)$, the algorithm of \cite{mt3} requires the bad-events to be monomials on $O(\log d)$ variables, and the algorithm of \cite{polylog-paper} requires the bad-events to involve $\polylog(n)$ variables. By contrast, the sequential deterministic and parallel randomized algorithms allow nearly arbitrary events.

Our   algorithm is based on a general methodology of Sivakumar \cite{sivakumar} for deterministic construction of probability spaces fooling certain types of ``statistical tests.''  It requires a structural property for the bad-events which is much less restrictive compared to previous algorithms: the bad-events must be computable by automata with roughly $\poly(d)$ states.  As some examples of our parallel algorithm, we get the following results:
\begin{theorem}
  \label{simp-thm2}
  \begin{enumerate}
    \item If $e \pmax d^{1+\eps} \leq 1$ for constant $\eps > 0$, and each bad-event can be determined by an automaton on $\poly(d)$ states, then we can find a good configuration in $NC^2$.
    \item If the vector of probabilities $p(B)^{1-\eps}$ satisfies Shearer's criterion for constant $\eps > 0$, and $\pmin \geq 1/\poly(n)$, and each bad-event $B$ can be determined by an automaton on $\poly(1/p(B))$ states, then we can find a good configuration in $NC^2$.
      \end{enumerate}
\end{theorem}

This provides a critical advantage for problems where the bad-events are based on sums of random variables. To illustrate,  consider a basic task in graph theory which we refer to as \emph{vertex-splitting}. In its simplest form, given a graph $G$ of maximum degree $\Delta$ we want to color the vertices red or blue so that each vertex gets at most $\tfrac{\Delta}{2} (1+\eps)$ neighbors of either color.   This is an important subroutine in constructions such as independent transversals (which we describe next) as well as defective vertex coloring (see e.g. \cite{polylog-paper}).  There are many extensions and variants on this procedure.

Vertex-splitting is easily handled by the LLL: we color the vertices randomly, and we have a bad-event that a vertex receives too many neighbors of either color. Bad-events for vertices $v,w$ can only affect each other if  they are within distance $2$ in the graph. Thus, $d = \Delta^2$.  Each bad-event can be computed by a simple automaton, which maintains a running counter of the number of the neighbors receiving each color. This has state space of $\Delta^2$, which indeed is $\poly(d)$. 

So this is suitable for our parallel derandomization algorithm in a straightforward way. However, this example, simple as it is, would not be covered by previous derandomization algorithms in general. For instance, the defective vertex coloring algorithm in \cite{polylog-paper} used an alternate LLL derandomization algorithm for vertex splitting, and consequently it required a number of preprocessing steps to reduce the original degree of the graph to $\polylog(n)$. Most LLL applications ``in the wild'' lacked deterministic parallel algorithms because of these types of limitations.

\bigskip

\noindent \textbf{The MT-distribution.} If the LLL applies, then a good configuration $X$ exists. In some applications, we need additional global information about it. One powerful tool  is the \emph{MT-distribution}  \cite{mtdist1}: namely, the distribution on $X$ at the termination of the MT algorithm. (This should not confused with the \emph{LLL-distribution}, namely, the distribution of $\Omega$ conditioned on all bad-events $\mathcal B$ being false, although the two distributions have similar properties.) The MT-distribution has many nice properties, as the probability of an event in the MT-distribution is often roughly comparable to its probability under the original distribution $\Omega$.

To derandomize the MT-distribution, we show that,  given a collection $\mathcal E$ of auxiliary events, we can find a good configuration for which a weighted sum of events in $\mathcal E$ is close to its expected value under the MT-distribution. (A formal statement requires many preliminary definitions.) We also show a parallel algorithm and results for more general LLL criteria. 

\subsection{Applications}
We provide a number of applications to combinatorial problems. Due to the generality of our algorithms, many of these are straightforward transplants of the existing combinatorial proofs. Section~\ref{app1-sec} describes a classic application of the asymmetric LLL to non-repetitive vertex coloring. We get the following result:
\begin{theorem}
  Let $\eps > 0$ be an arbitrary constant. For a graph $G$ of maximum degree $\Delta$, there is an $NC^2$ algorithm to obtain a non-repetitive vertex coloring of $G$ using $O( \Delta^{2+\eps})$ colors.
\end{theorem}

This almost matches the best non-constructive bounds \cite{mtdist2} which require $(1+o(1)) \Delta^2$ colors.

In Section~\ref{app2-sec}, we develop a more involved application to independent transversals. Given a graph $G$ along with a partition of its vertices into blocks of size $b$, an \emph{independent transversal (IT)} is an independent set of $G$ which contains exactly one vertex from each block. These objects have a long line of research, from both combinatorial and algorithmic points of view. The parallel MT algorithm can find an IT,  under the condition $b \geq (4+\eps) \Delta(G)$. A randomized sequential algorithm of Graf, Harris, \& Haxell \cite{it3} applies under the nearly-optimal bound $b \geq (2+\eps) \Delta(G)$.

Many combinatorial applications need  a \emph{weighted} IT. Specifically, given some vertex weighting $w: V \rightarrow \mathbb R$, we want to find an IT $I$ maximizing  $w(I) = \sum_{v \in I} w(v)$. This variant is also covered by the algorithm of \cite{it3}. For instance, this is critical for strong coloring (which we describe next).  Using our derandomization of the MT-distribution, along with a few other constructions, we show the following results:  
  \begin{theorem}
    \label{it-thm}
  Let $\eps > 0$ be an arbitrary constant, and let $G$ be a graph of maximum degree $\Delta$ with vertices partitioned into blocks of size $b$. 

  \begin{enumerate}
      \item If $b \geq (2+\eps) \Delta$, there is a deterministic poly-time algorithm to find an independent transversal $I$ of $G$ which additionally satisfies $w(I) \geq w(V)/b$.
    \item If $b \geq (4+\eps) \Delta$, there is an $NC^2$ algorithm to find an independent transversal $I$ of $G$ which additionally satisfies $w(I) \geq \Omega(w(V)/b)$.
  \end{enumerate}  
\end{theorem}

Finally, in Section~\ref{app3-sec}, we consider \emph{strong coloring}: given a graph $G$ partitioned into blocks of size $b$, we want to partition the vertex set into $b$ independent transversals. A randomized sequential algorithm under the condition $b \geq (3 +\eps) \Delta(G)$ is given in \cite{it3}, and a randomized parallel algorithm under the condition $b \geq 9.49 \Delta(G)$ is given in \cite{harris-llll-par}. Our results on weighted independent transversals give much more efficient deterministic algorithms in both these settings:
\begin{theorem}
  \label{it-thm2}
  Let $\eps > 0$ be an arbitrary constant, and let $G$ be a graph of maximum degree $\Delta$ with vertices partitioned into blocks of size $b$.
  \begin{enumerate}
      \item If $b \geq (3 + \eps) \Delta$, there is a deterministic poly-time algorithm for strong coloring of $G$.
  \item If $b \geq (5 + \eps) \Delta$, there is an $NC^3$ algorithm for strong coloring of $G$.
  \end{enumerate}  
\end{theorem}

\subsection{Limitations of our approach.}
The LLL is a very general principle with many extensions and variations. For completeness, we note some of the scenarios \emph{not} covered by our algorithm.
\begin{enumerate}
\item \textbf{The Lopsided Lov\'{a}sz Local Lemma}. Our definition of dependency is that $B$ is dependent with $B'$ if $\var (B) \cap \var(B') \neq \emptyset$. This can be relaxed to a slightly weaker relation known as \emph{lopsidependency}: $B$ is dependent with $B'$ when the two events \emph{disagree} on the value of some common variable. This provides correspondingly stronger bounds for applications such as $k$-SAT \cite{gst, harris-llll-seq}. Our general technique of enumerating witness dags may still work in this setting; however, this introduces a number of technical complications and we do not explore it here.

\item \textbf{Non-variable probability spaces.} A few applications of the LLL use probability spaces which are not defined by independent variables. For example, Erd\H{o}s \& Spencer \cite{erdos-spencer} showed the existence of certain types of  Latin transversals by applying the LLL for the uniform distribution on permutations. Such ``exotic'' probability spaces now have quite general efficient sequential \cite{harvey} and parallel \cite{harris-llll-par,ahk} randomized algorithms. No deterministic algorithms are known in these settings.

\item \textbf{Super-polynomial value of $m$ or $\sigma$.} We assume that the bad-event set $\mathcal B$ and the alphabet $\Sigma$ are provided explicitly as input. Thus, the algorithm may have runtime  polynomial in $m$ and $\sigma$. Efficient randomized algorithms are often possible given \emph{implicit} access to $\mathcal B$ and/or $\Sigma$. Specifically, for a given configuration $X \in \Sigma^n$, we need to determine which events in $\mathcal B$ are true, if any, and we need to sample variables from $\Omega$; see for example \cite{mtdist1, mtdist2, it3}. These procedures seem very difficult to derandomize in a generic way.
    \end{enumerate}
  
\subsection{Notation and Conventions}
We write $\tilde O(x) = x \polylog(x)$.  For integer $t$ we define $[t] = \{1, \dots, t \}$. 
For vectors $q_1, q_2$, we write $q_1 \leq q_2$ if $q_1(k) \leq q_2(k)$ for all coordinates $k$.

Our parallel algorithms are based on the EREW PRAM model, and we say an algorithm has $NC^k(\phi)$ complexity for some parameter $\phi$ if it uses $\poly(\phi)$ processors and $\tilde O(\log^k \phi)$ runtime. We say that an algorithm is in $NC^k$ if it has $NC^k(\phi)$ complexity where $\phi$ is the size of the input. The precise details of this computational model are not critical; briefly, it allows multiple processors with random read/write access to a single common memory, where any given cell can be accessed by at most one processor (either reading or writing) in any given time-step. 

 To avoid trivialities, we assume throughout that $p(B) \in (0,1)$ for all bad-events $B$. For any variable $k \in [n]$, we define $\mathcal B_k$ to be the set of bad-events $B \in \mathcal B$ with $k \in \var(B)$.

For a graph $G = (V,E)$ and a vertex $v \in V$, we define $N(v)$ to be the neighborhood of $v$, i.e. the set of vertices $u$ with $(u,v) \in E$. We define the maximum degree $\Delta(G) = \max_{v \in V} |N(v)|$; if $G$ is understood we write simply $\Delta$. For a vertex set $U \subseteq V$, we write $G[U]$ for the induced subgraph on $U$.

\section{Background on the Shearer criterion and the MT algorithm}
\label{shearer-sec}
The randomized MT algorithm is simple to describe (and hard to analyze). By contrast,  numerous definitions are required  simply to state the deterministic LLL algorithm. We begin with a self-contained overview of the Shearer criterion and its connection to the MT algorithm. We also describe how it relates to the more familiar criteria such as the symmetric and asymmetric  criteria. Most of the results in this section can be found in various forms in \cite{kolipaka, mt2, mt3, harvey}.

\subsection{Witness DAGs and the Shearer criterion}
Following \cite{mt3}, we define a \emph{witness DAG} (wdag) to be a finite directed acyclic graph $G$, where each vertex $v$ has a label $L(v) \in \mathcal B$, and which satisfies the additional condition that for all distinct vertices $v, v' \in G$ there is a directed edge between $v$ and $v'$ (in either direction) if and only if $L(v) \sim L(v')$.

 We let $\mathfrak W$ denote the set of all wdags. For a wdag $G$,  we define the \emph{size $|G|$} to be the number of nodes in $G$, the \emph{weight $w(G)$} to be $\prod_{v \in G} p(L(v))$, and the \emph{depth} to be the maximum path length of $G$.  For a collection of wdags $\mathfrak A \subseteq \mathfrak W$, we define $w(\mathfrak A) = \sum_{G \in \mathfrak A} w(G)$ and  $\maxsize(\mathfrak A) = \max_{G \in \mathfrak A} |G|$.

 We say a set $I \subseteq \mathcal B$ is \emph{stable} if $B \not \sim B'$ for all distinct pairs $B, B' \in I$. If a wdag $G$ has sink nodes $v_1, \dots, v_s$, then $L(v_1), \dots, L(v_s)$ are distinct and $\{ L(v_1), \dots, L(v_s) \}$ is a stable set of $\mathcal B$, which we denote by $\sink(G)$. For a stable set $I$, we let  $\mathfrak W_I$ be the set of all wdags $G$ with $\sink(G) = I$. 
 
With this notation, we can state the following criterion for the LLL:
\begin{definition}
  \label{ks-def}
The \emph{Shearer criterion holds} if  $w(\mathfrak W) < \infty$.
\end{definition}
 
If  the Shearer criterion holds, then $\Pr_{\Omega}( \bigcap_{B \in \mathcal B} \overline B) > 0$; in particular, a good configuration exists. This is the most general criterion for an abstract probability space with an appropriate dependency graph \cite{kolipaka, shearer}.\footnote{The original version of the Shearer criterion was formulated in terms of an object known as the \emph{independent-set polynomial} \cite{shearer}, and the presentation in \cite{kolipaka} is formulated in terms of an object known as a \emph{stable-set sequence}. Definition~\ref{ks-def} is equivalent to both of these.}   The Shearer criterion is also closely related to the algorithmic LLL. In this context, there are two types of wdags that play an especially important role in determining the ``history'' of a given bad-event $B$:

\begin{definition} We define $\mathfrak S_B$ to be the set of wdags $G$ with a single sink node labeled $B$, i.e.  $\mathfrak S_B =  \mathfrak W_{\{B \}}$. We define $\mathfrak C_B$ to be the set of wdags $G$ with $\sink(G) \subseteq \overline \Gamma(B)$, i.e. $\mathfrak C_B= \bigcup_{I \subseteq \overline \Gamma(B)} \mathfrak W_I$. (These wdags are ``collectible'' to $B$.) Equivalently, $\mathfrak C_B$ is the set of wdags obtained by removing the sink node from a wdag in $\mathfrak S_B$.

We also define $\mathfrak S = \bigcup_{B \in \mathcal B} \mathfrak S_B$ and $\mathfrak C = \bigcup_{B \in \mathcal B} \mathfrak C_B$.
\end{definition}

Note that, since $B \in \overline \Gamma(B)$,  we have $\mathfrak S_B \subseteq \mathfrak C_B$.  Also, in view of the relation between wdags in $\mathfrak S_B$ and $\mathfrak C_B$ (via removing/adding a sink node labelled $B$),  we define the fundamental parameter
$$
\mu(B) := w( \mathfrak S_B) = p(B) w(\mathfrak C_B) 
$$

The Shearer criterion, and associated bounds on $\mu$, are difficult to establish directly, so a number of simpler criteria are used in practice. We summarize a few common ones; the proofs appear in Appendix~\ref{shearer-sec-proof}.
\begin{proposition}
  \label{simple-shearer}
The Shearer criterion is satisfied under the following conditions:
  \begin{enumerate}
  \item (Symmetric LLL) If $e \pmax d \leq 1$. Furthermore, in this case $\mu(B) \leq e p(B)$ for all $B \in \mathcal B$.    
  \item If every $B \in \mathcal B$ satisfies $\sum_{A \in \Gamma(B)} p(A) \leq 1/4$.    Furthermore, in this case $\mu(B) \leq 4 p(B)$ for all $B \in \mathcal B$ with $\Gamma(B) \neq \emptyset$.
  \item (Asymmetric LLL) If there is a function $x: \mathcal B \rightarrow (0,1)$  which satisfies
    $$
    \forall B \in \mathcal B \qquad p(B) \leq x(B) \prod_{A \in \Gamma(B)} (1 - x(A))
    $$
    Furthermore, in this case $\mu(B) \leq \frac{x(B)}{1 - x(B)}$ for all $B \in \mathcal B$.
  \item (Bound by variables, c.f. \cite{beck}) If there is some $\lambda > 0$ which satisfies
    $$
    \forall k \in [n] \qquad \sum_{B \in \mathcal B_k} p(B) (1+\lambda)^{|\var(B)|} \leq \lambda.
    $$
    Furthermore, in this case $\mu(B) \leq (1+\lambda)^{|\var(B)|} p(B)$ for all $B \in \mathcal B$.   (Recall that $\mathcal B_k$ is the set of bad-events $B$ with $k \in \var(B)$.)
  \item (Cluster-expansion criterion \cite{bissacot}) If there is a function $\tilde \mu: \mathcal B \rightarrow [0,\infty)$ which satisfies
    $$
    \forall B \in \mathcal B \qquad  \tilde \mu(B) \geq p(B) \sum_{\text{stable $I \subseteq \overline \Gamma(B)$}} \prod_{A \in I} \tilde \mu(A)
      $$

      Furthermore, in this case $\mu(B) \leq \tilde \mu(B)$ for all $B \in \mathcal B$.
  \end{enumerate}
\end{proposition}

In analyzing the MT algorithm, it is often useful to consider a hypothetical situation where the bad-events have artificially inflated probabilities, while their dependency structure is left unchanged. For a vector $q: \mathcal B \rightarrow [0,1]$, we define the \emph{adjusted} weight of a wdag $G$ by $w_q(G) = \prod_{v \in G} q(L(v))$.  We likewise define $q_{\min} = \min_{B \in \mathcal B} q(B)$ and $q_{\max} = \max_{B \in \mathcal B} q(B)$ and $\mu_q(B) = w_q(\mathfrak S_B)$. 

We say  $q$ \emph{converges} if $w_{q}(\mathfrak W) < \infty$, and $q$ \emph{converges with $\eps$-slack} if $q^{1-\eps}$ converges for $\eps \in (0,1)$. (The vector $q^{1-\eps}$ is  defined coordinate-wise.)   We record a few useful bounds:
\begin{proposition}
  \label{wwbound}
If $p$ converges with $\eps$-slack, then for $q = p^{1-\eps/2}$, we have the following bounds:
  \begin{enumerate}
  \item $q$ converges with $\eps/2$-slack
  \item $\sum_{B \in \mathcal B_k} \mu_q(B) \leq \frac{2}{\eps (1-\pmax)}$ for each $k \in [n]$
  \item $\sum_{B \in \mathcal B} w_q(\mathfrak C_B) \leq \frac{2 n}{\eps \pmin (1-\pmax)}$
  \item $w_q(\mathfrak C_B) \leq (\frac{4}{\eps(1-\pmax)})^{|\var(B)|}$ for each $B \in \mathcal B$
 \end{enumerate}
\end{proposition}

\subsection{The resampling table}
In the MT algorithm as we have presented it, the new values for each variable are drawn in an online fashion when a bad-event $B$ is encountered. We refer to this as \emph{resampling} $B$. One key analytical technique of Moser \& Tardos \cite{mt} is to instead pre-compute a \emph{resampling table} $R$, which records an infinite list of values $R(i,0), R(i,1), \dots, $ for each variable $i$. Each variable $X(i)$ is initially set to $R(i,0)$; when it is first resampled, it is set to $R(i,1)$, and so on.  There is a natural probability distribution on $R$, which is to draw each entry $R(i,j)$ independently from the distribution of $X(i)$ in $\Omega$; with some abuse of notation, we refer to this as drawing $R \sim \Omega$.

 There are a number of important definitions to tie together the MT algorithm, the resampling table, and the analysis of wdags.
 
For a wdag $G$ and vertex $v \in G$, we define the configuration $X_{v,R}$ by $X_{v,R}(i) = R(i, j_i)$ for all $i$, where $j_i$ is the number of directed edges $(u,v) \in G$ with $L(u) \in \mathcal B_i$. For example, $X_{v,R}(i) = R(i,0)$ if $v$ is a source node of $G$.  We say $G$ is \emph{compatible} with $R$ if, for each node $v \in G$, the event $L(v)$ holds on  configuration $X_{v,R}$.  We define $\mathfrak R(R)$ to be the set of wdags $G$ which are compatible with $R$;  note that this depends only on entries $R(i,j)$ for $j \leq \depth(G)$. If $R$ is understood, we often simply write $\mathfrak R$.

For a wdag $G$ and a set $U$ of vertices of $G$, we define the wdag $G(U)$ to be the induced subgraph on all vertices $w$ with a path to any $u \in U$. We say a wdag $H$ is a \emph{prefix} of $G$, denoted $H \trianglelefteq G$, if $H = G(U)$ for a vertex set $U$. For a set of vertices $u_1, \dots, u_s$ we also write $G(u_1, \dots, u_s)$ as shorthand for $G(\{u_1, \dots, u_s \})$.
 
Given an execution of the MT algorithm up to some time $t$, we call the sequence of resampled bad-events $B_1, \dots, B_t$ an \emph{MT-execution on $R$};  the algorithm has not necessarily terminated at this point and some bad-events may still be true. We define an associated wdag we call the \emph{History-wdag} with vertices $v_1, \dots, v_t$ labeled $B_1,  \dots, B_t$, with a directed edge $(v_i, v_j)$ if $i < j$ and $B_i \sim B_j$.

We quote a few useful results from \cite{mt3}. For completeness, we include proofs in Appendix~\ref{shearer-sec-proof}.
\begin{proposition}
\label{gb1}
\begin{enumerate}
\item If $G$ is compatible with $R$ and $H \trianglelefteq G$, then $H$ is compatible with $R$.
\item The History-wdag is compatible with $R$.
\item For any wdag $G$, we have $\Pr_{R \sim \Omega}(\text{$G$ compatible with $R$}) = w(G)$.
\item For a fixed resampling table $R$, the MT algorithm performs at most $| \mathfrak S \cap \mathfrak R(R) |$ resamplings. Furthermore, it only uses entries $R(i,j)$ with $j \leq \maxsize ( \mathfrak S\cap \mathfrak R(R))$.
\end{enumerate}
\end{proposition}

At this point, let us briefly explain the roles played by the wdag sets $\mathfrak S$ and $\mathfrak C$ for the randomized and deterministic algorithms. Observe that when $R \sim \Omega$ the expected number of resamplings in the MT algorithm is at most
$$
\bE[ |\mathfrak S\cap \mathfrak R| ] = \sum_{G \in \mathfrak S} \Pr( \text{$G$ compatible with $R$} ) = \sum_{G \in \mathfrak S} w(G) = w(\mathfrak S)
$$

In particular, if Shearer's criterion holds, then $w(\mathfrak S) < w(\mathfrak W) < \infty$ and the resampling process terminates with probability one, and it has work factor approximately $w(\mathfrak S)$. This is the basic principle behind the randomized MT algorithm \cite{kolipaka}. For the deterministic algorithm, it is not enough to bound the \emph{expected} number of resamplings; any resampling, however unlikely, still must be explicitly checked. For a given bad-event $B$, the work-factor will be roughly $w(\mathfrak C_B)$ to check $B$, and $w(\mathfrak C)$ for the overall algorithm.

\subsection{The MT-distribution}
\label{mtdist-sec}
The \emph{MT-distribution} \cite{mtdist1} is the distribution on the configuration $X$ at the termination of the MT algorithm, assuming we have fixed some rule for which bad-event to resample at each time.  Consider some event $E$ in the probability space $\Omega$, which is a boolean function of a set of variables $\var(E)$. We introduce a number of related definitions:
\begin{itemize}
\item $\Gamma(E)$ is the set of bad-events $B$ with $\var(B) \cap \var (E) \neq \emptyset$.
  \item $\mathcal B^E \subseteq \mathcal B$ is the set of bad-events $B \in \mathcal B$ with $\Pr_{\Omega}(B \cap \neg E) > 0$.
  \item $\mathfrak C'_E$ is the set of wdags $G$ such that $\sink(G) \subseteq \Gamma(E)$ and $L(v) \in \mathcal B^E$ for all $v \in G$.
    \item $p(E) = \Pr_{\Omega}(E)$
    \item $\mu(E) = \Pr_{\Omega}(E) w(\mathfrak C'_E)$.
\end{itemize}

We record a few bounds; the proofs are analogous to Proposition~\ref{simple-shearer} and Proposition~\ref{wwbound} and are omitted.
\begin{proposition}
  \label{mtdist-simple1}
  \begin{enumerate}
\item If $e \pmax d \leq 1$, then $w(\mathfrak C'_E) \leq e^{e \pmax |\Gamma(E)|}$.
\item If $p$ converges with $\eps$-slack, then $w_q(\mathfrak C'_E) \leq (\frac{4}{\eps(1- \pmax)})^{|\var(E)|}$ for $q = p^{1-\eps/2}$.
\item If $\mathcal B^E$ satisfies Proposition~\ref{simple-shearer}(3) with function $x$, then $w(\mathfrak C'_E) \leq \prod_{B \in \Gamma(E)} (1 - x(B))^{-1}$.
\item If $\mathcal B^E$ satisfies Proposition~\ref{simple-shearer}(4) with parameter $\lambda$, then $w(\mathfrak C'_E) \leq (1+\lambda)^{|\var(E)|}$.
  \item  If $\mathcal B^E$ satisfies Proposition~\ref{simple-shearer}(5) with function $\tilde \mu$, then $w(\mathfrak C'_E) \leq \sum_{\text{stable $J \subseteq \Gamma(E)$}} \prod_{B \in J} \tilde \mu(B)$.
  \end{enumerate}
\end{proposition}

For any $G \in \mathfrak C'_E$, we define a configuration $X_{\text{root}, R}$ by setting $X_{\text{root}, R}(i) = R(i, j_i)$, where $j_i$ is the number of vertices $u \in G$ with $L(u) \in \mathcal B_i$. We say that $G$ is $E$-compatible with $R$ if $G$ is compatible with $R$ (in the usual sense as we have defined it), and in addition the event $E$ holds on configuration $X_{\text{root}, R}$ for $G$. We let $\mathfrak R'_E(R)$ denote the set of all wdags $G \in \mathfrak C'_E$ which are $E$-compatible with $R$; again, if $R$ is understood we just write $\mathfrak R'_E$. Since  $X_{\text{root}, R}$ follows the distribution $\Omega$ and it is determined by entries of $R$ disjoint from the configurations $X_{v,R}$, we have
$$
\Pr_{R \sim \Omega}(\text{$G$ is $E$-compatible with $R$}) = w(G) p(E)
$$

The following is the fundamental characterization of the MT-distribution; for completeness we include a proof in Appendix~\ref{shearer-sec-proof}.
\begin{theorem}[\cite{mtdist3}]
  \label{mtdist-thm1}
  Suppose that an MT-execution on resampling table $R$ produces a configuration $X$ on which $E$ is true. Then some $G \in \mathfrak C'_E$ is $E$-compatible with $R$. 
\end{theorem}

Theorem~\ref{mtdist-thm1} provides a simple and clean bound on the probability of event $E$ in the MT-distribution, namely:
\begin{equation}
\label{mtdistn-tt}
  \Pr_{MT}(E) \leq \sum_{G \in \mathfrak C'_E} \Pr(\text{$G$ is $E$-compatible with $R$})  = \sum_{G \in \mathfrak C'_E} w(G) p(E) =  \mu(E).
\end{equation}

So suppose we are given some set $\mathcal E$  of auxiliary events and non-negative weights $c_E$ for $E \in \mathcal E$.
By Eq.~(\ref{mtdistn-tt}), there exists a good configuration $X$ which additionally satisfies
\begin{equation}
  \label{mtdist-eqn2}
\sum_{E \in \mathcal E} c_E E(X) \leq \sum_{E \in \mathcal E} c_E \mu(E)
\end{equation}
where we write $E(X)$ as the indicator function, i.e. $E(X) = 1$ if $E$ holds on $X$, else $E(X) = 0$. 

By executing $O(1/\delta)$ independent repetitions of the MT algorithm, we can efficiently find  a good configuration satisfying the slightly weaker bound
\begin{equation}
  \label{mtdist-eqn3}
\sum_{E \in \mathcal E} c_E E(X)\leq (1+\delta) \sum_{E \in \mathcal E} c_E \mu(E)
\end{equation}
for any desired $\delta > 0$. We refer to a good configuration satisfying Eq.~(\ref{mtdist-eqn3}) as a \emph{$\delta$-good} configuration. Our goal will be to match this deterministically.

\section{Basic derandomization of MT}
\label{seq-sec}
Our overall strategy follows the same broad outline as \cite{mt, mt2}: when Shearer's criterion is satisfied with slack, then we search for a resampling table $R$ where $\mathfrak S\cap \mathfrak R$ has polynomial size (which must exist by probabilistic arguments), and then run the MT algorithm on $R$.   To complicate matters, $\mathfrak S$ is an infinite set;  we need to work with a family of wdags which allows us to ``finitely cover''  it.

Our algorithm will also handle the MT-distribution, so we suppose we are also given a set $\mathcal E$ of auxiliary events, with corresponding non-negative weights $c_E$ for $E \in \mathcal E$. We may have $\mathcal E = \emptyset$; many of our results will be specialized for this case. Assuming $\mathcal E$ and $\mathcal B$ are fixed, we define $$
\mathfrak F = \Bigl( \bigcup_{B \in \mathcal B} \mathfrak C_B \Bigr) \cup \Bigl( \bigcup_{E \in \mathcal E} \mathfrak C'_E \Bigr)
$$

For a threshold $\tau$ to be specified, we define ${\mathfrak F}^{\text{high}}_{\tau}$ to be the set of wdags $G \in \mathfrak F$ with $w_{p^{1-\eps}}(G) \geq \tau$, and ${\mathfrak F}^{\text{low}}_{\tau}$ to be the set of wdags $G \in \mathfrak F - \mathfrak F^{\text{high}}_{\tau}$ satisfying either of the two properties (i) $w_{p^{1-\eps}}(G) \geq \tau^2$, or (ii) $G$ has a single sink node $v$ and $w_q(G - v) \geq \tau$.  We also define $\mathfrak F_{\tau} = \mathfrak F^{\text{high}}_{\tau} \cup \mathfrak F^{\text{low}}_{\tau}$. 

Intuitively, ${\mathfrak F}^{\text{high}}_{\tau}$ and ${\mathfrak F}^{\text{low}}_{\tau}$ represent possible resamplings for the randomized algorithm which have high and low probabilities respectively. The remaining wdags in $\mathfrak F - \mathfrak F_{\tau}$ would represent resamplings with ultra-low probabilities; we will not need to deal with these explicitly.

 To bound the runtime of the algorithms in this setting, we define two main parameters:
$$
m' = | \mathcal E| + |\mathcal B |, \qquad  W_{\eps}  = \sum_{B \in \mathcal B} w_{p^{1-\eps}} (\mathfrak C_B) + \sum_{E \in \mathcal E} w_{p^{1-\eps}} (\mathfrak C'_E)
$$

\begin{lemma}
\label{gb3}
If ${\mathfrak F}^{\text{low}}_{\tau}\cap \mathfrak R = \emptyset$, then $\mathfrak F\cap \mathfrak R = {\mathfrak F}^{\text{high}}_{\tau}\cap \mathfrak R$.
\end{lemma}
\begin{proof}
We will show the contrapositive: if $\mathfrak F\cap \mathfrak R \neq {\mathfrak F}^{\text{high}}_{\tau}\cap \mathfrak R$, then the smallest such wdag $G \in (\mathfrak F\cap \mathfrak R) - \mathfrak F^{\text{high}}_{\tau}$ is in $F^{\text{low}}_{\tau}$.

If $s > 1$, then consider $H_1 = G(v_1)$ and $H_2 = G(v_2, \dots, v_s)$. These are both in $\mathfrak F$ and $\mathfrak R$. Since every node of $G$ is a node of $H_1$ or $H_2$ or both, we have $w_q(H_1) w_q(H_2) \leq w_q(G)$. Also, since $H_1$ omits $v_2$ while $H_2$ omits $v_1$, we have $|H_1| < |G|$ and $|H_2| < |G|$. By minimality of $|G|$, it must be that $H_1 \in \mathfrak F^{\text{high}}_{\tau}$ and $H_2 \in \mathfrak F^{\text{high}}_{\tau}$. So $w_q(H_1) \geq \tau$ and $w_q(H_2) \geq \tau$ which implies $w_q(G) \geq \tau^2$ and so $G \in \mathfrak F_{\tau}^{\text{low}}$.

If $s = 1$, then consider $H = G - v$. Here $H \in \mathfrak C_{L(v)} \subseteq \mathfrak F$. Since $H$ is a prefix of $G$, it is compatible with $R$. By maximality of $|G|$ we have $H \in \mathfrak F_{\tau}^{\text{high}}$ and so $w_q(G-v) \geq \tau$. Thus again $G \in \mathfrak F_{\tau}^{\text{low}}$.
\end{proof}

\begin{corollary}
\label{gb4}
Suppose that ${\mathfrak F}^{\text{low}}_{\tau}\cap \mathfrak R = \emptyset$. Then a necessary condition for $E$ to hold on the output of the MT algorithm is that ${\mathfrak F}^{\text{high}}_{\tau}\cap  \mathfrak R'_E \neq \emptyset$
\end{corollary}
\begin{proof}
If $E$ holds on the output, then by Theorem~\ref{mtdist-thm1} there exists some wdag $G \in \mathfrak R'_E$. In particular, $G \in \mathfrak F \cap \mathfrak R$. By Lemma~\ref{gb3}, we know that $G \in \mathfrak F^{\text{high}}_{\tau}$ as well.
\end{proof}

We summarize a few bounds on $\mathfrak F_{\tau}$ in terms of $W_{\epsilon}$, adapted and strengthened from the ``counting-by-weight'' arguments of \cite{mt2, mt3}. The proofs appear in Appendix~\ref{shearer-sec-proof}.
\begin{proposition}
  \label{gb2}
  Let $q = p^{1-\eps}$ for $\eps \in (0, \tfrac{1}{2})$ and let $\tau \in (0, \tfrac{1}{2})$. Then:
  \begin{enumerate}
  \item $W_{\eps} \geq \max\{ m', \frac{1}{1 - \pmax}, w_q(\mathfrak F) \}$
  \item $w( \mathfrak F^{\text{low}}_{\tau} ) \leq \tau^{\eps} W_{\eps}$
    \item $|\mathfrak F_{\tau}| \leq O(W_{\eps} / \tau^2)$
    \item $\maxsize(\mathfrak F_{\tau}) \leq O( \min \{ |\mathfrak F_{\tau}|,  W_{\eps} \log \tfrac{1}{\tau} \} )$
\end{enumerate}
\end{proposition}

Our algorithm requires a subroutine to enumerate wdags; since this is similar to prior works, we defer this to  Appendix~\ref{gb2a-proof}. The algorithm also has a problem-specific component; namely, it requires computing conditional probabilities for events in $\mathcal B$. Given any event $B \in \mathcal B$, and for $X \sim \Omega$, we need to compute the conditional probability of $B$ of the form:
$$
\Pr( \text{$B$ holds on $X$} \mid X(i_1) = j_1, \dots, X(i_k) = j_k )
$$

We call an algorithm to compute these quantities a \emph{partial-expectations oracle (PEO)}. Furthermore, the condition that a wdag $G$ is compatible with $R$ can be viewed as the conjunction of events $L(v)$ on the configurations $X_{v,R}$. We can compute conditional probabilities that $G$ is compatible with $R$, given that $R \sim \Omega$ and certain values of $R$ have been fixed, with $|G|$ executions of the PEO.\footnote{In some cases, we may not be able to exactly compute the conditional expectations of the bad-events, but we have some pessimistic estimator, that is, some statistic $S(X) \geq 0$ whose conditional expectations we can compute and where $S(X) \geq 1$ whenever $B$ holds. Our algorithms will still work in this case (except that the probability of $B$ would be replaced by $\mathbf E[S(X)]$).  To simplify our exposition, we do not explicitly discuss this variant.}

We are now ready to state our main derandomization algorithm:

\begin{theorem}
  \label{main-det-thm5}
  Suppose  $\mathcal B \cup \mathcal E$ has a PEO with runtime $T$.  Then there is a deterministic algorithm which takes an input parameter $\delta \in (0,1)$ and produces a $\delta$-good configuration, with runtime $(W_{\eps}/\delta)^{O(1/\eps)} n \sigma T $.
\end{theorem}
\begin{proof}
Let us define $\beta = W_{\eps}/\delta$.  We summarize our algorithm as follows:
  \begin{enumerate}
  \item Find a threshold $\tau$ with $w( \mathfrak F^{\text{low}}_{\tau}) \leq \delta/2$ and $| \mathfrak F_{\tau} | \leq \beta^{O(1/\eps)}$.
  \item Generate $\mathfrak F_{\tau}$.
  \item Select a resampling table $R$ to minimize a potential function $\Phi(R)$ defined in terms of $\mathfrak F_{\tau}$.
  \item Run the MT algorithm on $R$.
  \end{enumerate}
  
For steps (1) and (2), we use an exponential back-off strategy to set $\tau$. Namely, for $i = 0, 1,2, \dots, $ we use Theorem~\ref{enum-alg} to generate the wdag set $\mathfrak F_{\tau}$ where $\tau= 2^{-i}$, and then check if $w(\mathfrak F^{\text{low}}_{\tau}) \leq \delta/2$; if so, we stop the process.  Note that $\bar \tau = (\beta/4)^{-1/\eps}$ satisfies this condition, and so the process stops by iteration $\lceil \log_2 1/\bar \tau \rceil \leq O(  \frac{\log \beta}{\eps})$ with a final threshold $\tau \geq \bar \tau / 2$. By Proposition~\ref{gb2}(3), the resulting set $\mathfrak F_{\tau}$  has $|\mathfrak F_{\tau}| \leq O( W_{\eps}/\tau^2) \leq O( W_{\eps} / \bar \tau^2)  \leq \beta^{O(1/\eps)}$ and $\maxsize( \mathfrak F_{\tau}) = O( W_{\eps} \log \tfrac{1}{\tau} ) \leq O( \frac{ W_{\eps} \log \beta}{\eps})$. So the overall runtime for these steps is $\beta^{O(1/\eps)}$.

  For step (3), we compute $s = \sum_{E \in \mathcal E}  c_E p(E) w( \mathfrak F^{\text{high}}_{\tau} \cap \mathfrak C'_E)$, and then form the potential function
  $$
  \Phi(R) = | \mathfrak F^{\text{low}}_{\tau} \cap \mathfrak R(R) |+ \frac{1}{2 s} \sum_{E \in \mathcal E} | \mathfrak F^{\text{high}}_{\tau} \cap \mathfrak R'_E(R) |
  $$

Note that for $R \sim \Omega$ we have $\bE \bigl[ |\mathfrak F^{\text{low}}_{\tau}\cap \mathfrak R| \bigr] = w( \mathfrak F^{\text{low}}_{\tau}) \leq  \delta/2$, and similarly $\bE [ | \mathfrak F^{\text{high}}_{\tau} \cap \mathfrak R'_E(R) | ] = w( \mathfrak F^{\text{high}}_{\tau} p(E)$. So $\bE_{R \sim \Omega}[\Phi(R)] \leq \delta/2 + 1/2$.  We apply the method of conditional expectations, using our PEO, to find a value for the resampling table $R$ such that $\Phi(R)$ is at most is expectation, i.e. $\Phi(R) \leq \delta/2 + 1/2$. Since  $| \mathfrak F^{\text{low}}_{\tau} \cap \mathfrak R |$ is an integer and $\delta < 1$, this implies $\mathfrak F^{\text{low}}_{\tau}\cap \mathfrak R = \emptyset$.

 For step (4), we run the MT algorithm on $R$ in $n T \poly(\maxsize( \mathfrak F_{\tau}),m') \leq n T \beta^{O(1/\eps)}$ time, generating a good configuration $X$. By Corollary~\ref{gb4}, a necessary condition for $E$ to hold on $X$ is that some $G \in \mathfrak C'_E \cap \mathfrak F^{\text{high}}_{\tau}$ is $E$-compatible with $R$. We thus have
  $$
  \sum_{E \in \mathcal E} c_E E(X) \leq  \sum_{E \in \mathcal E} c_E | \mathfrak F^{\text{high}}_{\tau} \cap \mathfrak R'_E | \leq 2 s \Phi(R)
  $$
  
Since $\Phi(R) \leq \delta/2 + 1/2$, we have $\sum_{E \in \mathcal E} c_E E(X) \leq s \delta + s \leq (1+\delta) \sum_{E \in \mathcal E} c_E \mu(E)$.
\end{proof}

If we do not care about the MT-distribution, then setting $\mathcal E = \emptyset$ and $\delta = 1/2$ gives a simplified statement:
\begin{theorem}
  \label{main-det-thm1}
  Suppose  $\mathcal B$ has a PEO with runtime $T$.  Then there is a deterministic algorithm which runs in time $W_{\eps}^{O(1/\eps)} n \sigma T$ to find a good configuration.
\end{theorem}

Both Theorem~\ref{main-det-thm5} and Theorem~\ref{main-det-thm1} still have many parameters.   The following bounds apply to a few common situations and are easier to apply:
\begin{theorem}
  \label{main-thm4-summary}
  Suppose $\mathcal B \cup \mathcal E$ has a PEO with $\poly(n)$ runtime. Let $\eps > 0$ be an arbitrary constant, and define parameter $\phi = \max\{m', n, 1/\delta, \sigma \}$.
  
  \begin{enumerate}
  \item If $e \pmax d^{1+\eps} \leq 1$, then we can find a $\delta$-good configuration in  $\poly(\phi, e^{\pmax d'})$ time where $d' = \max_{E \in \mathcal E} |\Gamma(E)|$.
    
  \item If $\sum_{B \in \mathcal B_k} p(B)^{1-\eps} (1+\lambda)^{|\var(B)|} \leq \lambda$ for all $k \in [n]$, then we can find a $\delta$-good configuration in $\poly(\phi, (1+\lambda)^h)$ time where $h = \max_{E \in \mathcal B \cup \mathcal E} |\var(E)|$.

   \item If $p$ converges  with $\eps$-slack and $\pmax \leq 1 - \Omega(1)$, then we can find a $\delta$-good configuration in $\poly(\phi, e^h)$ time where $h = \max_{E \in \mathcal B \cup \mathcal E} |\var(E)|$.
  \end{enumerate}
\end{theorem}
\begin{proof}
  \begin{enumerate}
  \item Assume without loss of generality that $\eps < 1/4$. If $d \geq 2$, then $p^{1-\eps/4}$ satisfies the symmetric LLL criterion $e \pmax^{1-\eps/4} d \leq 1$. So $w_{p^{1-\eps}}(\mathfrak C_B) \leq e$ for all $B$. For each $E \in \mathcal E$, Proposition~\ref{mtdist-simple1} gives $w(\mathfrak C'_E) \leq e^{e \pmax d'}$. So overall $W_{\eps} \leq e m + m' e^{e \pmax d'}$. 
  
Similarly, if $d = 1$, then $p^{1-\eps/4}$ satisfies the cluster-expansion criterion with $\tilde \mu(B) = 2 p(B)$ for all $B$; in this case again $W_{\eps} \leq 2 m  + m' e^{e \pmax d'}$. 

  \item Here $p$ converges with $\eps$-slack. By Propositions~\ref{simple-shearer} and ~\ref{mtdist-simple1} we have $w(\mathfrak C'_E) \leq (1+\lambda)^{|\var(E)|}$ for all $E \in \mathcal E$ and $w(\mathfrak C_B) \leq (1+\lambda)^h$ for all $B \in \mathcal B$. So  $W_{\eps} \leq m'  (1+\lambda)^h$.

  \item Let $q = p^{1-\eps/2}$. By Proposition~\ref{wwbound}(4), we have $w_q(\mathfrak C_B) \leq (\frac{4}{\eps (1-\pmax)})^{|\var(B)|}$ for any bad-event $B$. Likewise, by Proposition~\ref{mtdist-simple1} we have $w_q(\mathfrak C'_E) \leq (\frac{4}{\eps (1-\pmax)})^{|\var(E)|}$ for any $E \in \mathcal E$. Thus,  $W_{\eps/2} \leq \frac{m'}{(\eps (1 - \pmax))^{O(h)}} \leq m' e^{O(h)}.$  \qedhere
  \end{enumerate}
\end{proof}

There is an additional useful preprocessing step when $\mathcal E = \emptyset$. We say a bad-event $B$ is \emph{isolated} if $\Gamma(B) = \emptyset$; we say a variable $i \in [n]$ is isolated if $i \in \var(B)$ for some isolated $B$. It is straightforward to use the PEO to determine values for isolated variables to make all isolated bad-events false. Using this preprocessing step and with simplification of parameters, we can obtain crisper bounds:
\begin{theorem}
  \label{main-thm1-summary}
  Suppose $\mathcal B$ has a PEO with $\poly(n)$ runtime. Let $\eps \in (0,1)$ be an arbitrary constant, and let $\phi = \max\{ m,n,\sigma \}$.
  \begin{enumerate}
  \item If $e \pmax d^{1+\eps} \leq 1$, then we can find a good configuration in  $\poly(\phi)$ time.
  
  \item If $\sum_{B \in \mathcal B_k} p(B)^{1-\eps} (1+\lambda)^{|\var(B)|} \leq \lambda$  for all $k \in [n]$, then we can find a good configuration in $\poly(\phi, (1+\lambda)^h)$ time where $h = \max_{B \in \mathcal B} |\var(B)|$.
  \item If $\sum_{A \in \Gamma(B) } p(A)^{1-\eps} \leq 1/4$ for all $B \in \mathcal B$, then we can find a good configuration in $\poly(\phi)$ time.

  \item If $p \leq q$ where the vector $q$ converges with $\eps$-slack, then we can find a good configuration in $\poly(\phi,1/\qmin)$ time.
    
  \end{enumerate}
\end{theorem}
\begin{proof}
  \begin{enumerate}
  \item[(1,2)] These are restatements of Theorem~\ref{main-thm4-summary}.
   
  \item[(3)]  Here, $p$ converges with $\eps$-slack and $w_{p^{1-\eps}}(\mathfrak C_B) \leq 4$ for all non-isolated bad-events $B$. Thus, the residual problem after removing the isolated bad-events has $W_{\eps} \leq 4 m$. 

  \item[(4)] Note that if $q(B) \geq 1 - q_{\text{min}}$ for any bad-event $B$, then $B$ must be isolated. For, if $B \sim B'$, we would have $q(B) + q(B') \geq 1-q_{\text{min}} + q_{\text{min}}^{1-\eps} > 1$. In particular, the Shearer criterion would fail on $q$ just taking into account the two bad-events $B, B'$.  So after removing isolated bad-events, we get $q(B) \leq 1 - q_{\text{min}}$ for all $B$. Now Proposition~\ref{wwbound}(3) gives $\sum_{B \in \mathcal B} w_{q^{1-\eps/2}}(\mathfrak C_B) \leq \frac{2 n}{\eps q_{\text{min}}^2}$. Since $p \leq q$, this implies $W_{\eps/2} \leq \frac{2 n}{\eps q_{\text{min}}^2}$. \qedhere
          \end{enumerate}
\end{proof}

We note one counter-intuitive aspect of Theorem~\ref{main-thm1-summary}(4): setting $q = p$ would give the result we have stated as Theorem~\ref{thm1x}. But this would imply that if some bad-event $B$ has $p(B) = 0$, then the algorithm would have infinite runtime. But this should only help us, as we can simply ignore $B$ in that case.  To explain this paradox, note that in  most applications of the LLL, we do not compute the exact probabilities $p(B)$; instead, we derive (often crude) upper bounds $q(B)$ on them. For example, in the symmetric LLL criterion, we set $q(B) = \frac{1}{e d}$ for all $B$. It is these \emph{upper bounds}, not the actual probabilities, that must be at least $1/\poly(n)$.

\section{Parallel algorithms via log-space statistical tests}
\label{stat-test-sec}
The algorithm of Section~\ref{seq-sec}, based on conditional expectations,  is inherently sequential. We will develop an alternate parallel algorithm, based on a general derandomization method of Sivakumar \cite{sivakumar} for fooling certain types of automata. Let us first provide an overview of this method, and next in Section~\ref{fool-mt-sec} we describe how to apply it to the MT algorithm.

Formally, we define an \emph{automaton} to be a tuple $(F, A,a^{\text{start}})$, where $A$ is a state space, $a^{\text{start}} \in A$ is the initial state, and $F: A \times [n] \times \Sigma \rightarrow A$ is the transition function. For brevity, we often write just $F$ for the automaton.  We define $\size(F)$ to be the cardinality of the state space $A$. We define multi-step transition functions $F^t: \Sigma^n \rightarrow A$ recursively by setting $F^0(X) = a^{\text{start}}$ and $F^{t}(X) = F( F^{t-1}(X), t, X(t))$ for $t \geq 1$.

Less formally, given an input sequence $X = (X(1), \dots, X(n))$, the automaton begins at the designated start state and updates its state $a$ as $a \leftarrow F( a, t, X(t))$ for $t = 1, \dots, n$. The transition function  here may depend on $t$; this is different from the usual definition in theory of formal languages.  To avoid technicalities, we always assume that $F$ can be computed in $NC^2(n)$.

Given an event $E$ determined by variables $X$, we say that automaton $F$ \emph{decides} $E$ if the state space of $F$ includes two terminal states labeled $0$ and $1$, such that $F^n(X) = E(X)$ for all $X \in \Sigma^n$. As a simple example, consider the event $\sum u_j X(j) \geq c$ for some given threshold $c$ and vector $u \in \{0, 1 \}^n$. This can be decided by an automaton whose state maintains the running sum $a = \sum_{j \leq t} u_j X(j)$. For $t < n$, it has transition function $F( a, t, x_t) = a + u_t X(t)$. For $t = n$, it  compares the running sum to the threshold $c$, updating $F(a, n, x_t)$ to be the indicator that $a+u_n X(n) \geq c$.

Such finite automata have surprisingly broad applications and powerful algorithmic tools. As two simple examples, we have the following:
\begin{observation}
  \label{simple-dynamic}
 Suppose automaton $F$ decides $E$ and has size $\eta$.
  \begin{enumerate}
  \item  In $\tilde O(\log n \log (n \sigma \eta))$ time and $\poly(n,\eta, \sigma)$ processors we can compute $\Pr_{\Omega}(E)$ for any given product distribution $\Omega$ over $\Sigma^n$.
  \item In $\tilde O(\log n \log (n \sigma \eta))$ time and $\poly(n,\eta, \sigma)$ processors, we can find a configuration $X$ avoiding $E$, if any such configuration exists.
  \end{enumerate}
\end{observation}
\begin{proof}
  For the first result, we recursively compute the probability of transiting from state $a_1$ at time $t$ to state $a_2$ at time $t + 2^h$, for all values $t = 0, \dots, n$ and all pairs of states $a_1, a_2$.  Taking $h = 0, \dots, \lceil \log_2 n \rceil$ contributes the $O(\log n)$ term in the runtime. For the second result,  we compute a configuration $X(t), \dots, X(t+2^h)$ which transits from any state $a_1$ at time $t$ to any state $a_2$ at time $t+2^h$ (if such exists). Again, enumerating over $h$ contributes $O(\log n)$ runtime.
\end{proof}

Note in particular that Observation~\ref{simple-dynamic} provides a PEO for the event $E$.

Nisan \cite{nis1, nis2} showed a much more powerful property of such automata, which is that they admit the construction of a ``fooling'' probability distribution $D$. We define this formally as follows:
\begin{definition}
Distribution $D$ fools the automaton $F$ to error $\eps$ if, for any state $s \in A$, we have
$$
\bigl| \Pr_{X \sim D} (F^n(X) = s) - \Pr_{X \sim \Omega}(F^n(X)=s) \bigr| \leq \eps
$$
\end{definition}

Nisan's original works \cite{nis1, nis2} did not give precise complexity bounds. Later work of \cite{mrs, harris2} further optimized the construction of $D$. We use the following slight variants of results of \cite{harris2}:

\begin{theorem}
  \label{fool-thm2}
Let $\Omega$ be a probability distribution on independent variables $X(1), \dots, X(n)$ over alphabet $\Sigma$ with $|\Sigma| = \sigma$. (The variables $X(i)$ may have different distributions.) Let $F_1, \dots, F_k$ be automata deciding events $E_1, \dots, E_k$. Define $\phi = \max\{ \size(F_i), n, k, \sigma\}$. 
  \begin{enumerate}
  \item There is a deterministic parallel algorithm to find a distribution $D$ of support size $\poly(\phi)$ which fools the automata  to any desired error $\eps > 0$. The algorithm has a complexity of  $\tilde O(\log (\phi/\eps) \log n)$ time and $\poly(\phi, 1/\eps)$ processors.
    \item There is a deterministic parallel algorithm which, for input parameters $\rho \in (0,1)$ and $s_1, \dots, s_k \geq 0$, produces a configuration $X \in \Sigma^n$ with $\sum_i s_i E_i(X) \leq (1+ \rho) \sum_i s_i \Pr_{\Omega} (E_i)$. The algorithm complexity is $\tilde O(\log (\phi/\rho) \log n)$ time and $\poly(\phi, 1/\rho)$ processors.
\end{enumerate}
\end{theorem}

The proofs use standard techniques and are deferred to Appendix~\ref{appendix2}.

In this context, the automata $F_i$ in this case are also called \emph{logspace statistical tests}, viewing the input variables $X(1), \dots, X(n)$ as an incoming data stream where each automaton computes some test statistic.  We next turn to applying this machinery to derandomize the MT algorithm.

\section{Logspace statistical tests for the Moser-Tardos algorithm}
\label{fool-mt-sec}
Now instead of requiring some (sequential) PEO, we suppose that each event $E \in \mathcal B \cup E$ can be decided by some automaton $F_E$.  We measure the complexity with the following definition:
\begin{definition}[Size bound for automata]
  \label{aut-complex-def}
The automata for $\mathcal B \cup E$ have \emph{size bound} $(r, \eta)$ if every $B \in \mathcal B$ has $\size(F_B) \leq \min\{ p(B)^{-r}, \eta \}$ and every $E \in \mathcal E$ has $size(F_E) \leq \eta$.
\end{definition}

By Observation~\ref{simple-dynamic}(1), we can efficiently compute $p(E)$ for any  $E \in \mathcal B \cup \mathcal E$.

The key algorithmic idea is to transform the automata $F_B$ into larger automata which decide if wdags are compatible with $R$. We will need to ensure that these larger automata read the entries of $R$ in the same order. We will use the \emph{lexicographic order} on the entries $R(i,j)$: i.e. the order $$
(1,0), (2,0), (3,0), \dots, (n,0), (1,1), \dots, (n,1), \dots.
$$

\begin{proposition}
  \label{eta-g-bound}
For any wdag $G$, there is an automaton $F_G$ to decide if $G$ is compatible with $R$, which has $\size(F_G) = \prod_{v \in G} \size(F_{L(v)})$. For any event $E$ and wdag $G$, there is an automaton $F_{E,G}$ to decide $G$ is $E$-compatible with $R$, which has $\size(F_{E,G}) = \size(F_E) \size(F_G)$.

These both read the entries $R(i,j)$ in the lexicographic order up to $j \leq \depth(G)$.
\end{proposition}
\begin{proof}
  For each node $v \in G$, the automaton $F_G$ has a state variable $a_v$ from state-space $A_{L(v)}$ of automaton $F_{L(v)}$. So $\size(F_G) = \prod_{v \in G} \size(F_{L(v)})$, and the state is a tuple $a = (a_v \mid v \in G)$. The starting state is $(a^{\text{start}}_v \mid v \in G)$ where $a^{\text{start}}_v$ is the starting state of automaton $F_{L(v)}$. When we process $R(i,j)$, we determine if there is any $v \in G$ such that $X_{v,R}(i)$ is determined to be $R(i,j)$. If there is such $v$ (necessarily unique), we update $a$ by updating $a_v \leftarrow F_{L(v)}(a_v, i, R(i,j))$.  At the end of the process, we get $G$ compatible with $R$ iff $a_v = 1$ for every $v \in G$.
  
The automaton $F_{E,G}$ also has an additional state variable $a_E$ from state-space $A_E$ of automaton $F_E$ which checks $X_{\text{root},R}$. The analysis is completely analogous.
\end{proof}

\begin{lemma}
  \label{eta-g-bound2}
  Suppose $p$ converges with $\eps$-slack for $\eps < 1/2$, and $\mathcal B \cup \mathcal E$ has automata with size bound $(r, \eta)$. Then the automata $F_G$ and $F_{E,G}$ in Proposition~\ref{eta-g-bound} read the entries $R(i,j)$ in lexicographic order up to $j \leq \maxsize( \mathfrak F_{\tau} )$ and have size at most $\eta^2 \tau^{-4 r}$.
\end{lemma}
\begin{proof}
Clearly each $G$ has $\depth(G) \leq \maxsize (\mathfrak F_{\tau}) - 1$. We next show the bound on $\size(F_G)$; the bound on $\size(F_{E,G})$ is completely analogous. First, if $w_q(G) \geq \tau^2$, then $w(G) \geq \tau^{\frac{2}{1-\eps}} \geq \tau^4$. Since $\size(F_B) \leq p(B)^{-r}$, we have
  $$
  \size(F_G) \leq \prod_{v \in G} \size(F_{L(v)}) \leq  \prod_{v \in G} p(L(v))^{-r} =  w(G)^{-r} \leq \tau^{-4 r}
  $$

Otherwise, if $G$ has a single sink node $u$ and $w_q(G-u) \geq \tau$, then since $\size(F_{L(u)}) \leq \eta$, we have
  \[
  \size(F_G) \leq \size(F_{L(u)}) \prod_{v \in G - u} \size(F_{L(v)})  \leq \eta \prod_{v \in G - u} p(L(v))^{-r} \leq \eta w(G - u)^{-r} \leq \eta \tau^{-4 r}. \qedhere
  \]
  \end{proof}

Using this subroutine, we get the following main result for our parallel algorithm:

\begin{theorem}
  \label{main-det-thm7}
  Suppose $p$ converges with $\eps$-slack, and $\mathcal B \cup \mathcal E$ has automata with size bound $(r, \eta)$. Then there is an algorithm taking an input parameter $\delta \in (0,1)$ and producing a $\delta$-good configuration, running in $\tilde O( \log( (W_{\eps}/\delta)^r n \eta \sigma) \log(W_{\eps} n / \delta) / \eps)$ time and $\poly( (W_{\eps} n/ \delta)^{r/\eps}, \eta, \sigma)$ processors.
\end{theorem}
\begin{proof}
We assume without loss of generality $\eps < 1/2$, and we write $\beta = W_{\eps}/\delta$ throughout.  We summarize the algorithm as follows:
  \begin{enumerate}
  \item Find  a threshold $\tau$ such that $w( \mathfrak F^{\text{low}}_{\tau}) \leq \delta/100$ and $| \mathfrak F_{\tau} | \leq \beta^{O(1/\eps)}$.
  \item Generate the set $\mathfrak F_{\tau}$.
  \item  Build automata to decide events whether wdags $G \in \mathfrak F_{\tau}$ are compatible with $R$.
  \item Select $R$ to minimize an appropriate potential function $\Phi(R)$ which is computed by a weighted sum of the automata.
  \item Simulate the MT algorithm on $R$.
  \end{enumerate}
  
  For steps (1) and (2), we use a double-exponential back-off strategy: for $i = 0, 1, 2, \dots$, we guess $\tau = 2^{-2^i}$, use Theorem~\ref{enum-alg} to generate $\mathfrak F_\tau$, and check if $w(\mathfrak F^{\text{low}}_{\tau}) \leq \delta/100$. Since $\bar \tau = (0.01 \beta)^{-1/\eps}$ satisfies this condition, the process terminates by iteration $\lceil \log_2 \log_2 (1/\bar \tau) \rceil$, with a final threshold $\tau \geq \bar \tau^2 \geq \beta^{-O(1/\eps)}$.   In each iteration, we have $\maxsize(\mathfrak F_{\tau}) = O(W_{\eps} \log \tfrac{1}{\tau}) \leq O(\frac{W_{\eps} \log \beta}{\eps})$ and $|\mathfrak F_{\tau}| \leq O(W_{\eps}/\tau^2) \leq \beta^{O(1/\eps)}$, so overall this uses $\beta^{O(1/\eps)}$ processors and $\tilde O(\tfrac{\log^2 \beta}{\eps})$ time. 
  
For step (3), we use Lemma~\ref{eta-g-bound2} to construct automata  to decide whether any $G \in \mathfrak F_{\tau}$ is compatible with $R$ or is $E$-compatible with $R$ for any event $E \in \mathcal E$. The automata have size $\tau^{-O(r)} \eta^{O(1)} \leq \poly(\beta^{r/\eps},\eta)$, and depend on the first $\maxsize (\mathfrak F_{\tau})$ rows of $R$.

For step (4), we first compute the values $a_E = p(E)  w(  \mathfrak F^{\text{high}}_{\tau} \cap \mathfrak C'_E )$ for each $E \in \mathcal E$, as well as the sums $s = \sum_{E \in \mathcal E} a_E c_E$ and $t = w( \mathfrak F^{\text{high}}_{\tau})$. We then form the potential function
  $$
  \Phi(R) =| \mathfrak F^{\text{low}}_{\tau} \cap \mathfrak R(R) |+ \frac{\delta}{100 t} | \mathfrak F^{\text{high}}_{\tau} \cap \mathfrak R(R) | +  \frac{1}{10 s}  \sum_{\substack{E \in \mathcal E}}c_E |  \mathfrak F^{\text{high}}_{\tau} \cap  \mathfrak R'_E(R) | 
  $$

As in the proof of Theorem~\ref{main-det-thm5},  the definitions of $s$ and $t$ and the condition on $\tau$ imply $\bE_{R \sim \Omega}[\Phi(R)] \leq 0.1 + 0.02 \delta$. 

Furthermore, $\Phi(R)$ is a sum of weighted indicator functions for if each $G$ is compatible with $R$ or is $E$-compatible with $R$; these in turn are decided by the automata. We thus apply Theorem~\ref{fool-thm2}(2) with parameter $\rho = 0.1 \delta$ for the statistic $\Phi(R)$. Observe that there are $k = |\mathfrak F_{\tau}|$ automata and each has size at most $\eta^2 \tau^{-4 r}$. So the complexity parameter $\phi$ of Theorem~\ref{fool-thm} satisfies $\phi \leq \poly( \beta, \tau^{-r}, \etamax, n, \sigma)$; overall this step uses $\poly(\beta^{r/\eps}, \eta, n, \sigma)$ processors and $\tilde O( \log(\beta^r n \eta \sigma) \log(\beta n)/ \eps)$ time. It produces a resampling table $R$ with
  $$
  \Phi(R) \leq (1+\rho) \bE_{R \sim \Omega} [\Phi(R)] \leq (1 + \rho) (0.1 + 0.02 \delta) \leq 0.1 + 0.1 \delta
  $$

As $\Phi(R) < 1$, we have $\mathfrak F^{\text{low}}_{\tau}\cap \mathfrak R  = \emptyset$ and hence $\mathfrak F\cap \mathfrak R \subseteq \mathfrak F^{\text{high}}_\tau$.   So $\mathfrak S\cap \mathfrak R$ is explicitly enumerated and has size $| \mathfrak S\cap \mathfrak R | \leq |\mathfrak F^{\text{high}}_{\tau} \cap \mathfrak R| \leq 100 t / \delta \times \Phi(R) \leq O(\beta)$. 

To finish, we use one additional optimization of \cite{mt3}: rather than executing the MT algorithm directly, we can simulate it efficiently via a single computation of a maximal independent set. As we show in Lemma~\ref{lem-post}, this runs in $\tilde O(\log^2(\beta n) / \eps)$ time and $\poly(\beta, 1/\eps, 1/\delta, n)$ processors. It generates a configuration $X$ which is the output of the MT algorithm on $R$; in particular, $X$ is good.   As $\mathfrak F^{\text{low}}_{\tau}\cap \mathfrak R  = \emptyset$,  Corollary~\ref{gb4} then gives:
 \begin{align*}
   \sum_{E \in \mathcal E} c_E E(X) &\leq  \sum_{\substack{E \in \mathcal E}} c_E | \mathfrak F_{\tau}^{\text{high}} \cap \mathfrak R'_E| \leq 10 s \Phi(R)   \leq s(1 + \delta)  \leq (1+\delta) \sum_{E \in \mathcal E} c_E \mu(E). \qedhere
\end{align*}
\end{proof}

For $\mathcal E = \emptyset$ we have the simpler result:
\begin{theorem}
  \label{main-det-thm3}
  Suppose $p$ converges with $\eps$-slack and $\mathcal B$ has automata with size bound $(r, \eta)$.    Then there is an algorithm with $\tilde O( \log(W_{\eps}^r n \eta \sigma) \log(W_{\eps} n)/\eps)$ time and $\poly( W_{\eps}^{r/\eps}, n,\eta, \sigma)$ processors to find a good configuration.
\end{theorem}

As before, the bounds of Theorem~\ref{main-det-thm7} and Theorem~\ref{main-det-thm3} can be simplified for a number of common LLL scenarios.
\begin{theorem}
  \label{main-thm4-summary2}
 Suppose $\mathcal B \cup \mathcal E$ has automata with size bound $(r, \eta)$ for an arbitrary constant $r \geq 1$. Let $\eps \in (0,1)$ be an arbitrary constant, and define parameter $\phi = \max\{ m', n, \sigma, \eta, 1/\delta \}$.
  
  \begin{enumerate}
  \item If $e \pmax d^{1+\eps} \leq 1$, then we can find a $\delta$-good configuration in $NC^2( \phi e^{\pmax d'})$ complexity, where $d' = \max_{E \in \mathcal E} |\Gamma(E)|$.
      \item If $\sum_{B \in \mathcal B_k} p(B)^{1-\eps} (1+\lambda)^{|\var(B)|} \leq \lambda$ for all $k \in [n]$, then we can find a $\delta$-good configuration in $NC^2( \phi (1+\lambda)^h)$ complexity where $h = \max_{E \in \mathcal B \cup \mathcal E} |\var(E)|$
  \item If $p$ converges with $\eps$-slack, then we can find a $\delta$-good configuration in $NC^2(\phi e^h)$ complexity where $h = \max_{E \in \mathcal B \cup \mathcal E} |\var(E)|$.
  \end{enumerate}
\end{theorem}
\begin{proof}
  \begin{enumerate}
  \item We have $w_{p^{1-\eps}}(\mathfrak C'_E) \leq e^{e \pmax |\Gamma(E)|} \leq e^{O(\pmax d')}$ for all $E \in \mathcal E$, and so $W_{\eps} \leq \poly(\phi)$.
  \item We have $w_{p^{1-\eps}}(\mathfrak C'_E) \leq (1+\lambda)^{|\var(E)|}$ for all $E \in \mathcal E$, and so $W_{\eps} \leq \poly(\phi)$.
  \item As in Theorem~\ref{main-thm4-summary}(3), we have  $W_{\eps/2} \leq m' (\eps(1 - \pmax))^{-O(h)}$.  We may assume that $\size(F_B) \geq 2$ for any bad-event  $B$, as otherwise we would have $p(B) \in \{0,1 \}$ and $B$ can simply be ignored. As a consequence of this, Definition~\ref{aut-complex-def} implies that $p(B) \leq 1 - \Omega(1)$ for all $B$. So $W_{\eps/2} \leq m' e^{O(h)}$. \qedhere
  \end{enumerate}
  \end{proof}

When $\mathcal E = \emptyset$,  we can use Observation~\ref{simple-dynamic}(2) to efficiently set all isolated variables, such that all the isolated bad-events become false. With this preprocessing step, and some simplification of the parameters, we get the following results:
\begin{theorem}
  \label{par-sum-easy2}
  \begin{enumerate}
  \item Suppose $p \leq q$ where the vector $q$ converges with $\eps$-slack for an arbitrary constant $\eps \in (0,1)$, and each $B \in \mathcal B$ is decided by automaton with size $\poly(1/q(B))$.  Then we can find a good configuration in $NC^2(m n \sigma / \qmin)$ complexity.
  \item   Suppose $e \pmax d^{1+\eps} \leq 1$ for an arbitrary constant $\eps > 0$, and each $B \in \mathcal B$ is decided by an automaton with size $\poly(d)$.  Then we can find a good configuration in $NC^2(m n \sigma)$ complexity.
  \end{enumerate}
\end{theorem}

\begin{theorem}
  \label{par-sum-thm}
  Let $r \geq 1, \eps \in (0,1)$ be arbitrary constants. Suppose $\mathcal B$ has automata with size bound $(r, \eta)$, and let $\phi = \max\{ m,n,\eta, \sigma \}$.
  
  \begin{enumerate}
  \item If $\sum_{A \in \Gamma(B)} p(A)^{1-\eps} \leq 1/4$ for all $B \in \mathcal B$, then we can find a good configuration in $NC^2(\phi)$ complexity.
  
  \item If $\sum_{B \in \mathcal B_k} p(B)^{1-\eps} (1+\lambda)^{|\var(B)|} \leq \lambda$ for all $k \in [n]$, then we can find a good configuration in $NC^2( \phi (1+\lambda)^h)$ complexity where $h = \max_{B \in \mathcal B} |\var(B)|$. 
  \end{enumerate}
\end{theorem}

\section{Non-repetitive vertex coloring}
\label{app1-sec}
As our first example application, we consider a classic and straightforward exercise for the LLL.  Given a graph $G$ with maximum degree $\Delta$, we want to $C$-color the vertices such that no vertex-simple path has a repeated color sequence, i.e. a path on distinct vertices $v_1, \dots, v_{2 \ell}$ receiving colors $c_1, \dots, c_{\ell}, c_1, \dots, c_{\ell}$ respectively. This problem was introduced by Alon et al. \cite{alon2}, based on old results of Thue for non-repetitive sequences. The minimum number of colors needed is referred to as the \emph{Thue number} $C=\pi(G)$.

The analysis in \cite{alon2} used the asymmetric LLL to show $\pi(G) \leq 2 e^{16} \Delta^2$. A series of works have since improved the constant factor and provided efficient randomized algorithms. Most recently, \cite{mtdist2} described a sequential poly-time zero-error randomized algorithm with $C = \Delta^2 + O(\Delta^{5/3})$ colors, as well as a zero-error randomized parallel algorithm in $O(\log^4 n)$ time with a slightly larger value $C = \Delta^2 + O(\frac{\Delta^2}{\log \Delta})$. (See \cite{wood} for a survey of bounds and algorithms for non-repetitive coloring.)

We get the following straightforward algorithmic derandomization:
\begin{theorem}
  Let $\eps > 0$ be an arbitrary constant. There is an $NC^2$ algorithm which takes as input a graph $G$ and returns a non-repetitive vertex coloring of $G$ using $O(\Delta(G)^{2 + \eps})$ colors.
\end{theorem}
\begin{proof}

 Let $C = \lceil \Delta^{2 + \eps} \rceil$, and define parameter $L = \lceil \frac{10 \log n}{\eps \log \Delta} \rceil$  . Consider applying the LLL, wherein each vertex selects a color from $[C]$ uniformly at random, and there are two types of bad-events. First, for each simple path $v_1, \dots, v_{2 \ell}$ where $\ell \leq L$, we have a bad-event that the vertices $v_1, \dots, v_{2 \ell}$ receive a repeated color sequence. Second, for each pair of paths $v_1, \dots, v_{L}, v'_1, \dots, v'_L$ with distinct vertices, we have a bad-event that path $v_1, \dots, v_L$ has the same color sequence as $v'_1, \dots, v'_L$; here $v_L$ does not need to be a  neighbor $v'_1$.

Observe that a good configuration corresponds to a non-repetitive vertex coloring. For, suppose a path $v_1, \dots, v_{2 \ell}$ has a repeated color sequence. If $\ell < L$, this directly corresponds to a bad-event. If $\ell \geq L$, then $v_1, \dots, v_{\ell}$ and $v'_1, \dots, v'_{\ell}$ are paths with distinct vertices and the same color sequence, where we define $v'_i = v_{\ell + i}$.
  
  We claim that, for $\Delta$ larger than some constant, the vector $p(B)^{1-0.1 \eps}$ satisfies the criterion of Proposition~\ref{simple-shearer}(4) with $\lambda = 1$.  For, consider the set of bad-events $\mathcal B_v$ which are affected by a given vertex $v$. For each $\ell = 1, \dots, L$ there are at most $\ell \Delta^{2 \ell - 1 }$ paths of length $2\ell$ going through $v$, and there are at most $n L \Delta^{2 \ell - 1}$ pairs of length-$L$ paths going through $v$. The former events have probability $C^{-\ell}$ and the latter have probability $C^{-L}$. If we define $\phi = 4 \Delta^2 / C^{1-0.1 \eps}$, we have 
  $$
  \sum_{B \in \mathcal B_v} p(B)^{1-0.1\eps} (1+\lambda)^{|\var(B)|} \leq \frac{n L \Delta^{2 L - 1} 2^{2 L} }{ C^{(1-0.1\eps) L} }  + \sum_{\ell=1}^{L-1} \frac{\ell \Delta^{2 \ell - 1} 2^{2 \ell}}{ C^{(1-0.1\eps) \ell}}  = \frac{n L}{\Delta} \phi^L + \frac{1}{\Delta} \sum_{\ell=1}^{L-1} \ell \phi^\ell.
  $$
   
   Since $C \geq \Delta^{2 + \eps}$ we have $\phi \leq \Delta^{-\eps/2} \leq 1/4$ for sufficiently large $\Delta$.   Our choice of $L$ gives $\phi^L \leq n^{-5}$, and $\sum_{\ell=1}^{L-1} \ell \phi^\ell \leq \sum_{\ell=1}^{\infty} \ell \phi^\ell = \frac{\phi}{(1 - \phi)^2} \leq 1/2$.
Thus, overall 
\[
\sum_{B \in \mathcal B_v} p(B)^{1-0.1 \eps} (1+\lambda)^{|\var(B)|} \leq \frac{n L}{\Delta} n^{-5} + \frac{1}{\Delta} (1/2)  \leq 1 = \lambda. 
\]

Now, there are $n^2 \Delta^{2 L-2}$ possibilities for the second of bad-event, and for each $\ell = 1, \dots, L$, there are $n \Delta^{2 \ell - 1}$ possibilities for the first type of bad-event. Summing over $\ell$,  the total number of bad-events is $\poly(n, \Delta^L)$; for fixed $\eps$ this is $\poly(n)$.
 
  We can construct automata for the bad-events by simply recording the colors taken by the vertices along each path. For a path of length $2 \ell$, this automaton has size $C^{2 \ell}$. The corresponding bad-event $B$ has probability $C^{-\ell}$.  Thus, $\size(F_B) \leq \min \{ p(B)^{-r}, \eta \}$ for $\eta = C^{2 L}, r = 2$, where note  that $C^{2 L} \leq \Delta^{(4+2 \eps) L} \leq \poly(n)$. Also, $|\var(B)| \leq 2L \leq O(\log n)$.  
  
  So applying Theorem~\ref{par-sum-thm}(2) with $\eps/10$ in place of $\eps$ gives the desired algorithm.
  \end{proof}

\section{Independent transversals}
\label{app2-sec}
As a more involved application, we examine the combinatorial structure known as the \emph{independent transversal}, which has attracted a long line of research. In this setting, we have a graph $G = (V, E)$ along with a partition $\mathcal V$ of $V$. The parts of this partition are referred to as \emph{blocks}. We say  the partition is $b$-regular if $|W| = b$ for all blocks $W$. An independent transversal (IT) of $G$ is a vertex set $I$ which is an independent set, and which additionally satisfies $|I \cap W| = 1$ for all $W \in \mathcal V$.

 Many combinatorial problems, such as graph list-coloring, can be formulated in terms of independent transversals; see \cite{graf} for a more extensive background. One fundamental problem is to determine sufficient conditions for the existence of  an independent transversal in a graph. For instance,  Haxell \cite{haxell-match} showed that bound $b \geq 2 \Delta$ suffices to guarantee existence of an IT, which by a matching lower bound of \cite{szabo} is tight. 

Some applications require a weighted generalization of IT's: given a vertex weighting $w:V \rightarrow \mathbb R$, we need to find an IT $I$ maximizing the sum  $w(I) = \sum_{v \in I} w(v)$.  Aharoni, Berger \& Ziv \cite{aharoni} showed that when $b \geq 2 \Delta$, there exists an IT $I$ with $w(I) \geq w(V) / b$. The work \cite{it3}, building on \cite{it2}, gives a nearly-matching randomized algorithm under the condition $b \geq (2 +\eps) \Delta$ for constant $\eps > 0$. Similarly, when $w(v) \geq 0$ for all $v$ (we say in this case that $w$ is \emph{non-negative}) and $b \geq 4  \Delta$, then \cite{harris-partial} shows that the randomized MT algorithm directly gives an IT $I$ with $w(I) \geq \Omega(w(V)/b)$. 

In this section, we derandomize these results. The constructions are based on using the LLL for vertex-splitting:  we sample a $q$ fraction of the vertices so that the blocks are reduced from size $b$ to roughly $q b$ and the vertex degrees are reduced to roughly $q \Delta$. We also use the MT-distribution to retain roughly $q w(V)$ of the vertex weights. 
\begin{lemma}
  \label{itlem1}
  There is an $NC^2$ algorithm which takes as input a graph $G$ with non-negative vertex weighting $w$, a $b$-regular vertex partition where $b \geq \max\{ K, \Delta(G) \}$ for some sufficiently large constant $K$, and a parameter $q$ with $\frac{\log^2 b}{b} \leq q \leq 1$,   and returns a vertex subset $L \subseteq V$ satisfying the following properties:
  \begin{enumerate}
  \item $q b - 10 \sqrt{q b \log b} \leq |L \cap W| \leq q b + 10 \sqrt{q b \log b}$ for all blocks $W$
  \item $|L \cap N(v)| \leq q D + 10 \sqrt{q D \log b}$ for all vertices $v$, where $D = \max \{b/10, \Delta(G) \}$
  \item $w(L) \geq q w(V) (1 - 1/b^8)$
  \end{enumerate}
\end{lemma}
\begin{proof}
For brevity we write $\Delta = \Delta(G)$ throughout.   We apply the LLL, wherein each vertex $v$ goes into $L$ independently with probability $q$. There are two types of bad-events: first, for each block  $W$, we have a bad-event that $| |L \cap W| - q b| \geq 10 \sqrt{qb \log b}$; second, for each vertex $v$, we have a bad-event that $|L \cap N(v)| > q b + 10 \sqrt{q D \log b}$. 

For the first type of event, note that $|L \cap W|$ is a binomial random variable with mean $\mu = q b$, and the bad-event would constitute a relative deviation of $\delta = 10 \sqrt{(\log b)/(q b)}$. By our bound on $q$, we have $\delta \leq 1$ for $K$ sufficiently large. Using a simplified version of Chernoff's bound, we get $\Pr_{\Omega}(B) \leq 2 e^{-\mu \delta^2/3} = 2 b^{-100/3}$. Similarly, for the second type of event, define $\hat \mu = q D$, and note that $|L \cap N(v) |$ is a binomial random variable with mean $q |N(v)| \leq q \Delta \leq \hat \mu$. The bad-event would constitute a relative deviation of $\delta = 10 \sqrt{\log b/(q D)}$ from $\hat \mu$; again we have $\delta \leq 1$ for $b$ sufficiently large and using the simplified form of Chernoff's bound, the probability is at most $e^{-\hat \mu \delta^2/3} \leq  e^{-q D \delta^2/3} = b^{-100/3}$.

So the LLL instance has $\pmax \leq 2 b^{-100/3}$. Each vertex affects one bad-event of the first type and $\Delta$ of the second type, while the first type of bad-event depends on $b$ vertices and the second type depends on $\Delta$ vertices. Since $\Delta \leq D \leq b$, this LLL instance has $d \leq 3 b^2$. So it satisfies the symmetric criterion $e \pmax d^{1+\eps} \leq 1$ for $\eps = 1/2$ and $b$ sufficiently large.

We will use the MT-distribution to obtain the bound on $w(L)$. Each vertex $v$ has an auxiliary event $E_v$  that $v \notin L$, with corresponding weight $c_{E_v} = w(v)$. As described above, the event $E_v$  affects at most $\Delta + 1 \leq 2 b$ bad-events. So Proposition~\ref{mtdist-simple1} gives $\mu(E_v) \leq \Pr_{\Omega}(E_v) e^{e \pmax |\Gamma(E_v)|} \leq (1-q) e^{e (b^{-33}) (2 b)} $; for $b$ sufficiently large, this is at most $(1 - q) (1 + b^{-31})$.

We next describe automata for the events. Each event $E_v$ can easily be decided by an automaton of size $2$. For the first type of bad-event, we need to compute the running sum of $|L \cap W|$; for the second type of bad-event, we need to compute the running sum of $|L \cap N(v)|$. These are determined by counters which are bounded in the range $\{0, \dots, b \}$ and $\{0, \dots, \Delta \}$ respectively. Since $\Delta \leq D \leq  b$, the size is at most $\poly(d)$.

At this point, we have all the ingredients to apply Theorem~\ref{main-thm4-summary2}(1) with $\delta = b^{-10}$. This has $NC^2(mn)$ complexity, and produces a vertex set $L$ satisfying the first two properties. Furthermore, the MT-distribution condition of Eq.~(\ref{mtdist-eqn3}) gives $\sum_{v \in V - L} w(v) \leq (1+\delta) \sum_{v \in V} w(v) \mu(E_v)$. So $$
w(L) = w(V) - \sum_{v \in V - L} w(v) \geq w(V) - (1 + b^{-10}) \cdot (1-q) (1 + b^{-31}) w(V).
$$
Since $q b \geq \log^2 b$, this is at least $q w(V) (1 - b^{-8})$ for $b$ sufficiently large.
\end{proof}

\begin{proposition}
  \label{itlem2}
 There is an $NC^2$ algorithm which takes as input a graph $G$ with a non-negative vertex weighting $w$, a $b$-regular vertex partition where $b \geq \max\{ K, \Delta(G) \}$ for some sufficiently large constant $K$, and returns a vertex subset $L'$ satisfying the following three properties:
  \begin{enumerate}
  \item The induced partition on $L'$ (that is, the partition $\{ W \cap L': W \in \mathcal V \}$) is $b'$-regular for $b' = \lceil \log^2 b - 10 \log^{3/2} b \rceil$
  \item $\Delta(G[L']) \leq b' \Bigl( (D/b)  + \log^{-0.4} b \Bigr)$ where $D = \max \{ b/10, \Delta(G) \}$.
  \item $w(L')/b' \geq (1 - \log^{-0.4} b) w(V) / b$
  \end{enumerate}
\end{proposition}
\begin{proof}
  We first apply Lemma~\ref{itlem1}  with $q = \frac{\log^2 b}{b}$, and for each block $W$ we discard the $|L \cap W| - b'$ vertices of lowest weight. Letting $L'$ denote the remaining vertices, each block $W$ has
  $$
  \frac{w(L' \cap W)}{b'} \geq \frac{w(L \cap W)}{|L \cap W|} \geq \frac{w(L \cap W)}{\log^2 b + 10 \log^{3/2} b}.
  $$

  Summing over all blocks $W$ and using the estimate $w(L) \geq q w(V) (1 - b^{-8})$ gives:
  $$
  \frac{w(L')}{b'} \geq \frac{w(V) (1 - b^{-8})}{\log^2 b + 10 \log^{3/2} b} \cdot \frac{\log^2 b}{b}.
  $$

  For $b$ sufficiently large, this is at least $w(V) (1 - \log^{-0.4} b) / b$. Also, each vertex $v$ has $|N(v) \cap L'| \leq |N(v) \cap L| \leq q D + 10 \sqrt{q D \log b}$, so:
  $$
\frac{ |N(v) \cap L'|}{b'} \leq \frac{ \tfrac{\log^2 b}{b}  D + 10 \sqrt{ \tfrac{\log^2 b}{b} D \log b} }{\log^2 b - 10 \log^{3/2} b}.
  $$

Simple analysis shows that for $D = \Theta(b)$ and $b$ sufficiently large, this is at most $D/b + \log^{-0.4} b$.
  \end{proof}

We now use a standard iterated LLL construction with repeated applications of Proposition~\ref{itlem2}.
\begin{proposition}
  \label{itlem3}
  Let $\eps, \lambda > 0$ be arbitrary constants and let $\phi \in (1,10)$. There is an $NC^2$ algorithm which takes as input a graph $G = (V,E)$ with a non-negative vertex weighting $w$ and a $b$-regular vertex partition for $b \geq (\phi + \eps) \Delta(G)$. It returns a vertex subset $L \subseteq V$ satisfying the following three properties:
  \begin{enumerate}
  \item The induced partition on $L$ is $b'$-regular for some $b' \leq O(1)$.
  \item $b' \geq  \phi \Delta(G[L])$
  \item $w(L)/b' \geq (1 - \lambda)  w(V)/b$
  \end{enumerate}
\end{proposition}
\begin{proof}
Our plan is to iteratively generate vertex sets $L_i$ for $i = 0,1,2, \dots, $,  where  $L_0 = V$ and we apply Proposition~\ref{itlem2} to the induced subgraph $G[L_i]$ to obtain the next vertex set $L_{i+1}$. We let $b_i$ be the block-size of the induced partition on $L_i$, and we stop this process when $b_i$ falls below some threshold $\tau$ ( to be specified).  The sequence $b_i$ then follows the recurrence relation:
\begin{eqnarray*}
b_0 = b, \qquad \qquad b_{i+1} = \lceil \log^2 b_i - 10 \log^{3/2} b_i \rceil
\end{eqnarray*}

Each iteration of Proposition~\ref{itlem2} takes $\tilde O(\log^2 n)$ time and this terminates by iteration $O(\log^* b)$. As we apply Proposition~\ref{itlem2}, we need to check the preconditions $b_i \geq K$ and $b_i \geq \Delta(G[L_i])$. The former condition is clear for $\tau$ sufficiently large. In order to check the latter condition, we will recursively define a parameter $\delta_i$ by 
$$
\delta_0 = b/(\phi + \eps), \qquad  \qquad \delta_{i+1} = b_{i+1} (\delta_i / b_i) + \log^{-0.4} b_i    
$$
and we will maintain the following condition at each step:
  \begin{equation}
\label{maintain-eqn}
\delta_i \geq \Delta(G[L_i])
 \end{equation}
  
By telescoping sums,  we can observe that $\delta_i/b_i =  1/(\phi + \eps)+ \sum_{j <i} \log^{-0.4} b_j$. The sequence $\log b_i$ decreases to zero super-exponentially, so for $\tau$ sufficiently large we can ensure that
\begin{equation}
  \label{x-eqn}
\phi \delta_i \leq b_i \leq (\phi +  \eps) \delta_i
\end{equation}
holds for all iterations $i$. Now let us turn to showing Eq.~(\ref{maintain-eqn}) and showing $b_i \geq \Delta(G[L_i])$ by induction on $i$. They are immediate for $i = 0$ by our hypotheses. For $i \geq 0$, we use the bounds of Eq.~(\ref{maintain-eqn}),(\ref{x-eqn}) to get $b_i \geq \phi \delta_i \geq \delta_i \geq \Delta(G[L_i])$.  For the next iteration $i$, we get:
  $$
  \Delta(G[L_{i+1}]) \leq b_{i+1} (D_i/b_i  + \log^{-0.4} b_i) \qquad  \text{for $D_i = \max \{ \Delta(G[L_i]), b_i/10 \}$}
  $$
  Because of the bound in Eq.~(\ref{x-eqn}), the bound $\phi \leq 10$, and the bound $\Delta(G[L_i]) \leq \delta_i$ from Eq.~(\ref{maintain-eqn}), we have $D_i \leq  \delta_i$ and hence $\Delta(G[L_{i+1}]) \leq \delta_{i+1}$ as desired.
  
 This process terminates at some iteration $J$, where the vertex set $L_J$ has blocks of size $b_J \leq \tau$ and $b_J \geq \phi \delta_j \geq \phi \Delta(G[L_J])$. Furthermore,  each iteration has $w(L_{i+1})/b_{i+1} \geq (1 - \log^{-0.4} b_i) w(L_i)/b_i$, so by telescoping products we have
$$
\frac{w(L_J)}{b_J} \geq \frac{w(V)}{b} \prod_{i=0}^{J-1} (1 - \log^{-0.4} b_J)
$$

For $\tau$ sufficiently large (as a function of $\lambda$),  this product  is at least $1-\lambda$.
\end{proof}

We are ready for the deterministic sequential algorithm for weighted IT's:
\begin{proof}[Proof of the sequential algorithm for Theorem~\ref{it-thm}]
A polynomial-time algorithm to find such an independent transversal $I$ is provided in \cite{it3}. Most of the steps in this algorithm are deterministic. The only randomized part is to solve the following task: for arbitrarily small constants $\eps, \lambda$, we are given a graph $G$ and $b$-regular vertex partition for $b \geq (2+\eps) \Delta$, with a non-negative vertex weighting $w$, and we must produce a vertex subset $L$ so the induced partition is $b'$-regular with $2 \Delta(G[L]) < b' \leq O(1)$, and which has $w(L)/b' \geq  (1 - \lambda) w(V)/b$. (See \cite[Lemma 18]{it3} for additional details.)   We can achieve this task deterministically by applying Proposition~\ref{itlem3} with $\phi = 2$. 
\end{proof}

We now turn to the parallel algorithm. We first discuss an $NC^2$ algorithm when $b$ is \emph{constant}.
\begin{proposition}
  \label{itlem0}
  Let $C$ be an arbitrary constant. There is an $NC^2$ algorithm which takes as input a graph $G$ with a non-negative vertex weighting $w$ and a $b$-regular vertex partition such that $4 \Delta(G) \leq b \leq C$, and returns an IT $I$ of $G$ with 
  $$
  w(I) \geq \Bigl( \frac{1-\lambda}{2 b - 1} \Bigr) w(V)
  $$
\end{proposition}
\begin{proof}
We apply the LLL, wherein each block $W$ independently selects a vertex $X_W$ uniformly at random. For each edge $(u,v)$ there is a bad-event that both end-points are selected.  We will apply the cluster-expansion criterion with $\tilde \mu(B) = \alpha$ for all $B$, for some scalar $\alpha > 0$ to be determined.  Consider an edge $(u,v)$ with associated bad-event $B$. To form a stable set $I \subseteq \overline \Gamma(B)$, we may select either $B$ itself, or we may select any other edge from the block of $u$ and any other edge from the block of $v$. There are at most $b \Delta - 1$ other edges in these blocks, so
  $$
  \sum_{\text{stable $I \subseteq \overline \Gamma(B)$}} \prod_{A \in I} \tilde \mu(A) \leq \alpha + (1 + (b \Delta - 1) \alpha )^2
  $$
Since $\Pr_{\Omega}(B) = 1/b^2$ and $\Delta \leq b/4$,  the cluster-expansion criterion is satisfied with $\eps$-slack if
  $$
  \alpha \geq (1/b^2)^{1-\eps} \bigl( \alpha + (1 + (b^2/4 - 1) \alpha)^2 \bigr),
  $$
and it is routine to verify this holds for some $\alpha, \eps > 0$ which depends solely on $b$.

Following \cite{mtdist3}, we use the MT-distribution to bound the resulting weight $w(I)$. For each vertex $v$ in a block $W_v$, define the event $E_v$ that $v \notin I$, with associated weight $c_{E_v} = w(E_v)$.  Let us compute $\mu(\mathfrak C'_{E_v})$. Note that $\mathcal B^{E_v}$  does not contain bad-events for any edge $(u,v')$ where $v' \in W_v - \{ v \}$.  We apply the cluster-expansion criterion for the bad-events $\mathcal B^{E_v}$; for a bad-event corresponding to an edge $(u_1, u_2)$ where $u_1, u_2 \in V - W_v$, we set $\tilde \mu(B) = \alpha_1 = 4/b^2$, and for a bad-event corresponding to an edge $(u,v)$ we set $\tilde \mu(B) = \alpha_2 = \frac{4}{2 b (b-1)}$. For the first type of bad-event, a stable set $I \subseteq \overline \Gamma(B)$ may have one event from the block of $u_1$ and another from the block of $u_2$. Since $\alpha_2 \leq \alpha_1$, we calculate:
  $$
p(B)   \sum_{\text{stable $I \subseteq \overline \Gamma(B)$}} \prod_{A \in I} \tilde \mu(A) \leq b^{-2} (1 + b \Delta \alpha_1 )^2 = 4/b^2 = \tilde \mu(B)
$$

  For the second type of event, a stable set $I \subseteq \overline \Gamma(B)$  may have one event from an edge $(v,u')$ and a second from the block of $u$, giving
  $$
p(B)    \sum_{\text{stable $I \subseteq \overline \Gamma(B)$}} \prod_{A \in I} \tilde \mu(A) \leq b^{-2} (1 + \Delta \alpha_2 ) (1 + b \Delta \alpha_1) = \frac{4 b}{b^2(2 b-1)} = \tilde \mu(B)
    $$

Now, clearly $\Pr_{\Omega}(E_v) = (b-1)/b$, and stable-set $I \subseteq \Gamma(E_v)$ has at most one edge $(v,u')$, so 
    $$
    \mu(\mathfrak C'_{E_v}) \leq \frac{b-1}{b} \Bigl( 1 + \Delta \alpha_2 \Bigr) \leq \frac{2b-2}{2b-1}
    $$
    
    The events $E_v$ and bad-events $B$ can easily be checked by automata of size 2. We apply Theorem~\ref{main-thm4-summary2}(3) with $\delta = \lambda/(2b - 2)$, noting that the bad-events as well as the events in $\mathcal E$ involve at most $2$ variables. This generates an independent transversal $I$ with
  \begin{align*}
    \sum_{v \notin I} w(v)  \leq (1 + \delta) \sum_v c_{E_v} \mu(E_v) \leq w(V) \Bigl( 1 + \frac{\lambda}{2b-2} \Bigr) \Bigl( \frac{2b-2}{2b-1} \Bigr) = w(V) \Bigl( \frac{2 b - 2 + \lambda}{2b - 1} \Bigr)
  \end{align*}
and so
  \[
  w(I) = w(V) - \sum_{v \notin I} w(v) \geq w(V) - w(V) \Bigl( \frac{2 b - 2 + \lambda}{2b - 1} \Bigr)= w(V) \Bigl( \frac{1 - \lambda}{2 b - 1} \Bigr). \qedhere
  \]
\end{proof}

We can now obtain the parallel part of Theorem~\ref{it-thm}. We show a slightly more detailed result:
\begin{theorem}
  Let $\eps, \lambda > 0$ be arbitrary constants. There is an $NC^2$ algorithm which takes as input a graph $G = (V,E)$ with a non-negative vertex weighting $w$ and a $b$-regular vertex partition where $b \geq (4 + \eps) \Delta(G)$. It returns an independent transversal $I$ of $G$ with $w(I) \geq (\tfrac{1}{2}-\lambda) w(V)/b$.
 \end{theorem}
\begin{proof}
Apply Proposition~\ref{itlem3} with $\phi = 4$. For fixed $\eps$ and $\lambda$, the resulting graph $G' = G[L]$ has block-size $b'$ with $4 \Delta(G') \leq b' \leq O(1)$, and has $w(L')/b' \geq (1 - \lambda) w(V)/b$.   Next apply Proposition~\ref{itlem0} to $G'$, getting an IT $I$ with $w(I) \geq (1 - \lambda) w(L')/(2 b' - 1)$. Since $\lambda$ here is an arbitrary constant, the result follows by rescaling $\lambda$.
\end{proof}

\section{Strong coloring}
\label{app3-sec}
We now consider strong coloring of graphs, which is closely related to independent transversals. Given a $b$-regular vertex partition $\mathcal V$ of a graph $G = (V,E)$, we define a \emph{strong coloring} of $G$ with respect to $\mathcal V$ to be a partition of $V$ into independent transversals $V = I_1 \sqcup \dots \sqcup I_b$. Equivalently, it is a proper vertex $b$-coloring where exactly one vertex in each block receives each color.

The best general bound currently known \cite{haxstrong2} in terms of $b$ and $\Delta$ is that strong coloring is possible when $b \geq (11/4) \Delta$ for large $\Delta$.  Aharoni, Berger, \& Ziv \cite{aharoni} showed that the condition $b \geq 3 \Delta$ also suffices (for all $\Delta$); this was turned into a randomized sequential algorithm in \cite{it3} under the slightly stronger condition $b \geq (3 + \eps) \Delta$ for arbitrary constant $\eps > 0$.  It is conjectured that the correct bound is $b \geq 2 \Delta$.

The parallel algorithms lag behind significantly. The current best randomized algorithm \cite{harris-llll-par} is based on the Lopsided LLL for random permutations; it requires $b \geq (\frac{256}{27} + \eps) \Delta$ and runs in $\tilde O(\log^4 n)$ time. We are not aware of any deterministic parallel algorithms.

Let us explain the construction of \cite{aharoni}: consider a  \emph{partial strong coloring}, that is, a partial coloring where there is no edge $(u,v)$ where both vertices both receive the same (non-blank) color, and likewise there is no pair of vertices in the same block with the same color. We let $U(\chi)$ denote the number of \emph{uncolored} vertices in $\chi$.

 For a given index $i \in [b]$, we may form an \emph{augmentation graph} $G^{\text{aug}}_i$ from $G$, where for each  $i$-colored vertex $v$, we remove all vertices in the block of $v$ which have the same color as a neighbor of $v$.  Given an IT $I$ of $G^{\text{aug}}_i$, we form a new partial coloring $\chi_{\text{aug}}$ as follows: for each block $W$ with vertex $v \in I \cap W$, set $\chi_{\text{aug}}(v) = i$ and $\chi_{\text{aug}}(v') = \chi(v)$ for the vertex $v' \in \chi^{-1}(i) \cap W$ (if any). All other vertices $x$ have $\chi_{\text{aug}}(x) = \chi(x)$. The fundamental observation of \cite{aharoni} is the following:
\begin{proposition}[\cite{aharoni}]
  \label{gc1}
  The graph $G^{\text{aug}}_i$ has blocks of size at least $b - \Delta$. Furthermore,  $\chi_{\text{aug}}$ is a partial strong coloring with $U(\chi_{\text{aug}}) \leq U(\chi) - s$, where $s$ is the number of vertices $v \in I$ such that $v$ was uncolored in $\chi$ and no other vertex in the same block as $v$ was colored $i$ under $\chi$ 
\end{proposition}

We can immediately obtain a deterministic sequential algorithm.
\begin{proof}[Proof of the sequential part of Theorem~\ref{it-thm2}]
  We begin with the empty partial coloring, and apply a series of augmentation steps: namely, we choose an arbitrary uncolored vertex $u$, and let $i$ be some color not appearing in its block. We discard arbitrary vertices from $G^{\text{aug}}_i$, aside from vertex $u$, so that each block has size exactly $b' = b - \Delta$; let $V'$ denote the remaining vertices in $G^{\text{aug}}_i$.   Since $b' \geq (2 + \eps) \Delta$, we  apply Theorem~\ref{it-thm} with the weight function defined by $w(v) = 1$ if $v = u$, otherwise $w(v) = 0$. This yields an IT $I$ with $w(I) \geq w(V')/b' = 1/b' > 0$, and hence $u \in I$. By our choice of $u$, Proposition~\ref{gc1} gives $U(\chi_{\text{aug}}) \leq U(\chi) - 1$.     In particular, we get a full strong coloring after $n$ rounds.
\end{proof}

We now turn to the parallel algorithm. There are two stages to the construction: first, we get an algorithm for $5 \Delta(G) \leq b \leq O(1)$; we next use vertex-splitting to extend to unbounded $b$.
\begin{theorem}
  \label{pb1}
  Let $C$ be an arbitrary constant. There is an $NC^2$ algorithm which takes as input a graph $G = (V,E)$ with a $b$-regular vertex partition $\mathcal V$ such that $5 \Delta(G) \leq b \leq C$, and returns a strong coloring of $G$ with respect to $\mathcal V$.
\end{theorem}
\begin{proof}
  We produce  $\chi$ through a series of augmentation steps, starting with the initial coloring $\chi_0$ being the empty coloring. Let us consider some stage $\ell$ of the process, where we have a partial strong coloring $\chi_{\ell}$. Let $A$ be the set of uncolored vertices in $\chi_{\ell}$; we say a block $W$ \emph{omits} color $j$ if if $\chi_{\ell}^{-1}(j) \cap W = \emptyset$.  For each color $i$, we define the quantity 
 $$
  \Phi_{i}  = \sum_{W \in \mathcal V: \text{$W$ omits $i$}} | A \cap W |,
  $$
  and note that if we sum over all possible colors $j = 1, \dots, b$ we get:
  $$
  \sum_j \Phi_{ j} = \sum_{W \in \mathcal V} | A \cap W | \sum_{j: \text{$W$ omits $j$}} 1 = \sum_{W \in \mathcal V} | A \cap W |^2.
  $$
    
  Now let us select a color $i^*$ to maximize the quantity $\Phi_{i^*}$.   For this color $i^{*}$, define the weight function $w$ by $w(v) = 1$ if $v \in A$ and the block of $v$ omits color $i^*$, and $w(v) = 0$ otherwise; note that $w(V) = \Phi_{i^*}$.  We form graph $G^{\text{aug}}_{i^*}$ and discard the lowest-weight $b - b'$ vertices from each block, where $b' = b - \Delta$. We then apply Proposition~\ref{itlem0} to $G^{\text{aug}}_{i^*}, w$ with some constant parameter $\lambda$, and augment $\chi_{\ell}$ by the resulting IT $I$ to obtain the next coloring $\chi_{\ell+1}$.  
  
  By our definition of $w$ and Proposition~\ref{gc1}, we have $U(\chi_{\ell+1}) \leq U(\chi_{\ell}) - w(I)$. It remains to estimate $w(I)$.   Only colored vertices are removed from $G$ to get $G^{\text{aug}}_{i^*}$, and we discard the lowest-weight vertices to reduce the blocksize to $b'$, so overall $w(V') \geq \Omega(w(V))$ where $V'$ is the vertex-set of $G^{\text{aug}}_{i^*}$. Since $w(V) = \Phi_{i^*}$ and $b$ is constant we have:
  $$
  w(I) \geq w(V') (1-\lambda)/(2 b' - 1) \geq \Omega(w(V)) \geq \Omega( \Phi_{ i^*} ).
  $$
  
  Since $i^*$ is chosen as the maximum, we have $\Phi_{i^*} \geq \sum_i \Phi_i / b \geq |A \cap W |^2/b$. Since $b$ is constant, we thus have $w(I) \geq \Omega( \sum_W |A \cap W|^2 ) \geq \Omega(U(\chi_{\ell}))$. So we have shown that
   $$
 U(\chi_{\ell+1}) \leq U(\chi_{\ell}) - \Omega( U( \chi_{\ell} )).
  $$

Thus $U(\chi_{\ell}) = 0$ for some $\ell = O(\log n)$, so that $\chi_{\ell}$ is a full strong coloring. Each iteration requires finding a weighted IT via Proposition~\ref{itlem0}, which requires $\tilde O(\log^2 n)$ time.
\end{proof}

We now extend to large block-size via a series of vertex-splitting steps.
\begin{proposition}
  \label{itlem4}
  There is an $NC^2$ algorithm which takes as input a graph $G = (V,E)$ with a $b$-regular vertex partition where $b \geq K$ for some sufficiently large constant $K$, and returns disjoint vertex sets $L^{(1)}, L^{(2)}$ such that $V = L^{(1)} \cup L^{(2)}$ and satisfying the following  properties:
  \begin{enumerate}
  \item The induced partition on $L^{(1)}$ is $\lceil b/2 \rceil$-regular.
  \item The induced partition on $L^{(2)}$ is $\lfloor b/2 \rfloor$-regular.
  \item For $i = 1,2$ we have $\Delta(G[L^{(i)}]) \leq D/2 + 20 \sqrt{b \log b}$ where $D = \max \{ b/10, \Delta(G) \}$.
  \end{enumerate}
\end{proposition}
\begin{proof}
We first use the LLL to partition $V$ into sets $A^{(1)}, A^{(2)}, A^{(12)}$ with the following properties:
 \begin{enumerate}
  \item For all blocks $W \in \mathcal V$ we have $| W \cap A^{(i)} | \leq b/2$ for $i = 1,2$
   \item For all vertices $v$ we have $|N(v) \cap ( A^{(i)} \cup A^{(12)} )| \leq D/2 + 20 \sqrt{b \log b}$
   \end{enumerate}
   
Specifically, we place each vertex $v$ into exactly one of the sets $A^{(1)}, A^{(2)}, A^{(12)}$, independently, with probabilities $p, p, 1 - 2 p$ respectively, for $p = 1/2 - 10 \sqrt{ \frac{b}{\log b}}$. (For $K$ sufficiently large we have $p \in [0,1/2]$ so this is a valid probability distribution.) One can show that this random process satisfies the symmetric LLL criterion with $\eps$-slack and the bad-events can be decided via automata on $\poly(d)$ states. The analysis is very similar to Lemma~\ref{itlem1}, so we omit the proof.

   Given these sets $A$, we form $L^{(1)}, L^{(2)}$ by starting with $L^{(1)} = A^{(1)}, L^{(2)} = A^{(2)}$; for each block $W$, we move $\lfloor b/2 \rfloor - | W \cap A^{(1)} |$ vertices of $A^{(12)} \cap W$ to $L^{(1)}$ and the remaining vertices in $A^{(12)} \cap W$ to $L^{(2)}$. The first condition on the sizes of $W \cap A^{(1)}, W \cap A^{(2)}$ ensures we have enough vertices in $A^{(12)} \cap W$ to allow this. Also, since $L^{(i)} \subseteq A^{(i)} \cup A^{(12)}$ the second condition on $|N(v) \cap ( A^{(i)} \cup A^{(12)} )|$ ensures the bound on  $\Delta(G[L^{(i)}])$ for $i = 1,2$.
\end{proof}

We can now complete the proof of Theorem~\ref{it-thm2}:
\begin{proof}[Proof of the parallel part of Theorem~\ref{it-thm2}]
  Given the graph $G = (V,E)$ with a vertex partition $\mathcal V$, we will apply a series of vertex-splitting steps via Proposition~\ref{itlem4}; at each stage $i = 0, $ we have disjoint vertex sets $V_{i,j}$ where $j = 0, \dots, 2^i - 1$, with the following properties:
  \begin{enumerate}
  \item For each $i$, the sets $V_{i,0}, \dots, V_{i,2^i-1}$ partition $V$. 
  \item Each induced partition on $V_{i,j}$ is $b_{i,j}$-regular where $\lfloor b/2^i \rfloor \leq b_{i,j} \leq \lceil b/2^i \rceil$.
    \item Each set $V_{i,j}$ has $\Delta(G[V_{i,j}]) \leq \delta_i$, for a parameter $\delta_i$ to be determined.
  \end{enumerate}

  Initially, we set $V_{0,0} = V$, and we apply this process as long as $b/2^i \geq \tau$ for some constant threshold value $\tau$ (for fixed $\eps$). This terminates after $O(\log b)$ rounds. In each round, we apply Proposition~\ref{itlem4} in parallel to induced subgraphs $G[V_{i,j}]$, and we set $V_{i+1, 2 j}, V_{i+1,2 j+1}$ to be the resulting sets $L^{(1)}, L^{(2)}$, with corresponding block-sizes $b_{i+1,2j} =  \lceil b_{i,j}/2 \rceil$ and $b_{i+1, 2j+1} = \lfloor b_{i,j} / 2 \rfloor$. 
  
  For $\tau$ sufficiently large, we ensure that $b_{i,j} \geq K$, and clearly the bounds $\lfloor b/2^i \rfloor \leq b_{i,j} \leq \lceil b/2^i \rceil$ hold.   The overall process terminates in $\tilde O(\log^3 n)$ time. The next step is determine the values $\delta_i$. We define them via the recurrence
$$
\delta_0 = \frac{b}{4+\eps}, \qquad \delta_{i+1} = \delta_i/2 + 21 \sqrt{ (b/2^i) \log(b/2^i)}
$$

We claim that  $b_{i,j} \geq 4 \delta_i$ for all $i$. For this, we first show by induction on $i$ that $\delta_{i+1} \leq \delta_0/2^i + (b/2^i)^{3/4}$. The base case $i = 0$ clearly holds. For the induction step, we have
\begin{align*}
  \delta_{i+1} &= \delta_i/2 + 21 \sqrt{ (b/2^i) \log(b/2^i)} \leq ( \delta_0/2^i + (b/2^i)^{3/4} ) /2 + 21 \sqrt{ (b/2^i) \log(b/2^i)}
  \end{align*}
and simple analysis shows this is at most $\delta_0/2^{i+1} + (b/2^{i+1})^{3/4}$ for $b/2^i \geq \tau$ and $\tau$ sufficiently large. Accordingly, to show $b_{i,j} \geq 4 \delta_i$, it therefore suffices to show that
$$
b/2^i -1 \geq 4( \delta_0/2^i + (b/2^i)^{3/4})  = (b/2^i)/(1+\eps/4) + 4 (b/2^i)^{3/4}
$$
which is clear for $b/2^i \geq \tau$ and $\tau$ is sufficiently large.

We can now show the crucial property that $\Delta(G[V_{i,j}]) \leq \delta_j$
for all $i,j$ by induction on $i$. Again, the base case is immediate. For the induction step, Proposition~\ref{itlem4} gives $\Delta(G[V_{i+1,j}]) \leq D_i/2 + 20 \sqrt{b_{i,j} \log b_{i,j}}$. Since $b_{i,j} \leq b/2^i + 1$, and $b/2^i \geq K$ for a sufficiently large constant $K$, this is at most $D/2 + 21 \sqrt{(b/2^i) \log (b/2^i)}$.  Observe that $\delta_i \geq \delta_0 / 2^i \geq \frac{b}{(4+\eps) 2^i}$ and $b_{i,j} \leq \lceil b/2^i \rceil \leq b/2^i + 1$. Since clearly $b/2^i + 1 \leq \frac{10 b}{(4+\eps) 2^i}$ holds when $b/2^i \geq \tau$ and $\tau$ is sufficiently large, this shows  $b_{i,j} \leq 10\delta_i$ and hence $D_i \leq \delta_i$. Thus indeed $\Delta(G[V_{i+1,j}]) \leq  \delta_{i+1}$.

Thus the preconditions of Proposition~\ref{itlem4} hold at each iteration $i$. At the end of this process, we apply Theorem~\ref{pb1} in parallel to the graphs $G[V_{i,j}]$, and we combine all the resulting strong colorings $\chi_{i,j}$ into a single strong coloring $\chi$ on $G$. Each $G[V_{i,j}]$ is $b_{i,j}$-regular where $b_{i,j} \leq 2 \tau + 1 \leq O(1)$, so these all run in $\tilde O(\log^3 n)$ time.
\end{proof}

\section{Acknowledgments}
Thanks to Aravind Srinivasan for helpful discussions, and for pointing me to the method of Sivakumar.  Thanks to Penny Haxell and Alessandra Graf for some discussions about derandomization in the context of independent transversals. Thanks to David Wood for finding some errors in the calculations for non-repetitive coloring. Thanks to Jonas H{\"u}botter and journal reviewers for finding a number of corrections and errors.

\appendix
\section{Calculations for the Shearer criterion and wdags}
\label{shearer-sec-proof}

\begin{proof}[Proof of Proposition~\ref{simple-shearer}]  
We have stated the criteria from least general  to most general, but we prove them in the reverse direction.
\begin{enumerate}
\item[(5)]  For $h \geq 0$, let $\mathfrak W^h_I$ be the set of wdags in $\mathfrak W_I$ with depth at most $h$. We show by induction on $h$  that $w(\mathfrak W^h_I) \leq \prod_{B \in I} \tilde \mu(B)$ for all $I$. The base case $h = 0$ is immediate. For the induction step, let $I = \{B_1, \dots, B_t \}$. Given a wdag $G \in \mathfrak W_I^h$, removing the sink nodes from $G$ yields a wdag $G'$ of depth at most $h-1$. Furthermore, the stable set $J = \sink(G')$ can be partitioned (not necessarily uniquely) as $J = \bigcup_{i=1}^t J_i$, where $J_i \subseteq \overline \Gamma(B_i)$. Thus, we have:
$$
w(\mathfrak W^h_I) \leq \prod_{B \in I} p(B) \negthickspace \negthickspace \negthickspace \negthickspace \negthickspace \sum_{\substack{\text{disjoint stable-sets} \\
      J_1 \subseteq \overline \Gamma(B_1), \dots, J_t \subseteq \overline \Gamma(B_t)}} \negthickspace \negthickspace \negthickspace  \negthickspace \negthickspace w(\mathfrak W^{h-1}_{J_1 \cup \dots \cup J_t})
      $$
      By applying the induction hypothesis to each term $ w(\mathfrak W^{h-1}_{J_1 \cup \dots \cup J_t})$, we then get:
      \begin{align*}
w(\mathfrak W^h_I)   &\leq \prod_{B \in I} p(B) \negthickspace \negthickspace  \sum_{\substack{\text{stable sets} \\ 
      J_1 \subseteq \overline \Gamma(B_1), \dots, J_t \subseteq \overline \Gamma(B_t)}}   \prod_{B \in J_1 \cup \dots \cup J_t} \negthickspace \negthickspace \negthickspace \tilde \mu(B)  \\
  &\leq \sum_{\text{stable $J_1 \subseteq \overline \Gamma(B_1)$}} p(B_1) \prod_{A_1 \in J_1} \tilde \mu(A_1)  \thickspace \thickspace \cdots \negthickspace \negthickspace \sum_{\text{stable $J_t \subseteq \overline \Gamma(B_t)$}} p(B_t) \prod_{A_t \in J_t} \tilde \mu(A_t) \\
  &\leq \tilde \mu(B_1) \cdots \tilde \mu(B_t) \qquad \text{by hypothesis}  
\end{align*}

Taking the limit as $h \rightarrow \infty$ gives $w(\mathfrak W_I) \leq \prod_{B \in I} \tilde \mu(B) < \infty$ for all $I$, so Shearer's criterion is satisfied. For $I = \{B \}$, we have $\mu(B) = w(\mathfrak W_{ \{B \} }) \leq \tilde \mu(B)$.

\item[(4)] We apply the cluster-expansion criterion with $\tilde \mu(B) = p(B) (1+\lambda)^{|\var(B)|}$ for each $B$. To enumerate the stable sets $I \subseteq \overline \Gamma(B)$, we may choose, for each $k \in \var(B)$, at most one other bad-event $A \in \mathcal B_k$ to place into $I$. So we calculate
$$
  \sum_{\text{stable $I \subseteq \overline \Gamma(B)$}} \prod_{A \in I} \tilde \mu(A) \leq \prod_{k \in \var(B)} \Bigl(1 + \sum_{A \in \mathcal B_k} \tilde \mu(A) \Bigr) \leq \prod_{k \in \var(B)} \Bigl(1 + \sum_{A \in \mathcal B_k} p(A) (1 + \lambda)^{|\var(A)|}  \Bigr)$$

By hypothesis, this is at most $(1+\lambda)^{|\var(B)|}$. So $p(B) \sum_{\text{stable $I \subseteq \overline \Gamma(B)$}}  \prod_{A \in I} \tilde \mu(A) \leq p(B) (1+\lambda)^{|\var(B)|} = \tilde \mu(B)$  as required for the cluster-expansion criterion. Furthermore, we have $\mu(B) \leq \tilde \mu(B) = p(B) (1+\lambda)^{|\var(B)|}$.

\item [(3)] Following \cite{harvey}, we apply the cluster-expansion criterion with $\tilde \mu(B) = \frac{x(B)}{1 - x(B)}$ for each $B$.  Here, $$
\sum_{\text{stable $I \subseteq \overline \Gamma(B)$}}  \prod_{A \in I} \tilde \mu(A) \leq \sum_{I \subseteq \overline \Gamma(B)} \prod_{A \in I} \tilde \mu(A) =  \prod_{A \in \overline \Gamma(B)} (1 + \tilde \mu(A)) =  \prod_{A \in \overline \Gamma(B)} (1 - x(A))^{-1}.
$$
    Thus, using the given bound $p(B) \leq x(B) \prod_{A \in \Gamma(B)} (1 - x(A))$, we have
  \begin{align*}
    p(B) \negthickspace  \negthickspace \negthickspace \negthickspace \negthickspace \sum_{\text{stable $I \subseteq \overline \Gamma(B)$}}  \prod_{A \in I} \tilde \mu(A) &\leq \Bigl( x(B) \negthickspace \negthickspace \prod_{A \in \Gamma(B)}\negthickspace \negthickspace (1 - x(A)) \Bigr) \Bigl( \negthinspace \prod_{A \in \overline \Gamma(B)} \negthickspace  \negthickspace (1 - x(A))^{-1} \Bigr) = \frac{x(B)}{1 - x(B)} = \tilde \mu(B).
  \end{align*}
as required for the cluster-expansion criterion.

\item[(2)] We apply the asymmetric criterion with $x(B) = \min \{1, 2 p(B) \}$ for all $B$. If $\Gamma(B) = \emptyset$, the condition is trivially satisfied. Otherwise, suppose $\Gamma(B) \neq \emptyset$. Note that, letting $B'$ be an arbitrary element of $\Gamma(B)$, we have $p(B) \leq \sum_{A \in \Gamma(B')} p(A) \leq 1/4$. Now  $$
x(B) \prod_{A \in \Gamma(B)} (1 - x(A))  \geq 2 p(B) \prod_{A \in \Gamma(B)} (1 - 2 p(A))  \geq 2 p(B) (1 - 2  \sum_{A \in \Gamma(B)}p(A));
$$
this is at least $2p(B) (1 - 2 \cdot 1/4) = p(B)$ by hypothesis, as required by the asymmetric criterion. Furthermore,  we have $\mu(B) \leq \frac{x(B)}{1 - x(B)} \leq \frac{2 p(B)}{1 - 2 p(B)} \leq \frac{2 p(B)}{1 - 2 \times 1/4} = 4 p(B)$.
\item We apply the asymmetric criterion  with $x(B) = \frac{d e p(B)}{d+1}$ for all $B$. We have:
  \begin{align*}
    x(B) \prod_{A \in \Gamma(B)} (1 - x(A)) &=  \frac{d e p(B)}{d+1} \prod_{A \in \Gamma(B)} (1 - \frac{d e p(A)}{d+1}) \geq \frac{d e p(B)}{d+1} \Bigl( 1 - \frac{d e \pmax}{d+1} \Bigr)^{d-1} \\
    &\geq p(B) \times \frac{d e}{d+1} \Bigl( 1 - \frac{1}{d+1} \Bigr)^{d-1} \geq p(B).
  \end{align*}

Furthermore, $\mu(B) \leq \frac{x(B)}{1 - x(B)} = \frac{e d p(B)}{d+1 - e d p(B)} \leq \frac{e d p(B)}{d+1 - e d \pmax} \leq e p(B)$.   \qedhere
\end{enumerate}
\end{proof}

Before proving Proposition~\ref{wwbound}, we recall a useful result. It was originally shown in \cite{mt2}, but we provide a new and simpler proof here.
\begin{proposition}[\cite{mt2}]
  \label{gv5}
    If $p(1 + \eps)$ converges, then $\sum_{B \in \mathcal B_k} \mu(B) \leq 1/\eps$ for any $k \in [n]$.
\end{proposition}
  \begin{proof}
Let us define $\mathfrak A = \bigcup_{B \in \mathcal B_k} \mathfrak S_B$, and define $\mathfrak A' \subseteq \mathfrak A$ to be the set of wdags $G \in \mathfrak A$ with the additional property that $L(v) \notin \mathcal B_k$ for all \emph{non-sink} vertices $v$.  We can construct a function $F$ mapping finite sequences  of $\mathfrak A'$ into $\mathfrak A$ as follows. Consider $G_1, \dots, G_{\ell} \in \mathfrak A'$ with sink nodes $v_1, \dots, v_{\ell}$ respectively. Define $G = F(G_1, \dots, G_{\ell})$ by taking copies of $G_1, \dots, G_{\ell}$, along with edges from $v \in G_i$ to $v' \in G_j$ if $i < j$ and $L(v) \sim L(v')$. Note that $G$ is a wdag and has only a single sink vertex $v_{\ell}$, since any $x \in G_i$ has a path to $v_i \in G_i$ which in turn has an edge to $v_{\ell}$.

To see that this is a bijection, take $G \in \mathfrak A$ and let $U$ denote the set of vertices $u \in G$ with $L(u) \in \mathcal B_k$. These vertices must be linearly ordered, so suppose they are sorted as $U = \{ u_1, \dots, u_{\ell} \}$. Then $(G_1, \dots, G_{\ell})$ is the unique pre-image of $G$ where we define inductively $G_i = G(u_i) - G_1 - \dots - G_{i-1}$.
  
The nodes of $F(G_1, \dots, G_{\ell})$ are the union of those of $G_1, \dots, G_{\ell}$, so $w_{q}(F(G_1, \dots, G_{\ell})) = w_{q}(G_1) \cdots w_{q}(G_{\ell})$ for any vector $q$. Summing over $\ell \geq 0$ and $G_1, \dots, G_{\ell} \in \mathfrak A'$ on the one hand, and over all $G \in \mathfrak A$ on the other, we have:
  \begin{equation}
    \label{gm2}
  w_{q}(\mathfrak A) = \sum_{\ell=1}^{\infty} \sum_{G_1, \dots, G_{\ell} \in \mathfrak A'} w_{q}(G_1) \cdots w_{q}(G_{\ell}) =  \sum_{\ell=1}^{\infty} w_{q}(\mathfrak A')^{\ell}
  \end{equation}

By hypothesis, $w_{p(1+\eps)}(\mathfrak A) \leq w_{p(1+\eps)}(\mathfrak W) < \infty$. By Eq.~(\ref{gm2}) with $q = p(1+\eps)$ this implies that $w_{q}(\mathfrak A') < 1$, and hence:
  $$
  w(\mathfrak A') = w_{q/(1+\eps)}(\mathfrak A') \leq \frac{ w_{q}(\mathfrak A')}{1+\eps} < \frac{1}{1+\eps}
  $$

  Applying Eq.~(\ref{gm2}) to $q = p$ then gives:
  \[
  \sum_{B \in \mathcal B_k} \mu(B) = w(\mathfrak A) = \sum_{\ell=1}^{\infty} w(\mathfrak A')^{\ell} < \sum_{\ell=1}^{\infty} (1+\eps)^{-\ell} = 1/\eps. \qedhere
  \]
\end{proof}

\begin{proof}[Proof of Proposition~\ref{wwbound}]
  \begin{enumerate}
      \item Here $q^{1-\eps/2} = p^{1 - \eps + \eps^2/4} \leq p^{1-\eps}$, and  $p^{1-\eps}$ converges.
\item  Let $\nu = \eps (1 - \pmax)/2$. Simple analysis shows $q (1+\nu) \leq p^{1-\eps}$; since $p^{1-\eps}$ converges, $q(1+\nu)$ converges as well. The result follows by applying Proposition~\ref{gv5} to $q$ and $\nu$.
\item  Since $w_q(\mathfrak C_B) = \mu_q(B) / q(B) \leq \mu_q(B) / \pmin^{1-\eps/2} \leq \mu_q(B) / \pmin$, result  (2) gives
$$
\sum_{B \in \mathcal B} w_q(\mathfrak C_B) \leq  \sum_{k \in [n]} \sum_{B \in \mathcal B_k} w_q(\mathfrak C_B) \leq \frac{2 n}{\eps \pmin(1 - \pmax)}
$$

\item We have $w_q(\mathfrak C_B) =\sum_{\text{stable $J \subseteq \overline \Gamma(B)$}}  w(\mathfrak W_J)$. We can form a stable set $J$ by going through each $k \in \var(B)$, and choosing at most $A \in \mathcal B_k$ to place into $J$. From  part (2) we  estimate
    \[
w_q(\mathfrak C_B) \leq \prod_{k \in \var(B)} \Bigl( 1 + \sum_{A \in \mathcal B_k} \mu_q(A) \Bigr)  \leq \Bigl( 1+\frac{2}{\eps (1- \pmax)}\Bigr)^{|\var(B)|} \leq \Bigl( \frac{4}{\eps (1-\pmax)} \Bigr)^{|\var(B)|}.  \qedhere
    \]
\end{enumerate}    
\end{proof}

\begin{proof}[Proof of Proposition~\ref{gb1}]
  \begin{enumerate}
  \item Consider any vertex $v \in H$. Since $H \trianglelefteq G$, any vertex $w \in G$ with an edge to $v$ is also in $H$. As a result, $H(v) = G(v)$, so the configuration $X_{v,R}$ is the same for both $G$ and $H$. Since $G$ is compatible with $R$, also $H$ is compatible with $R$.
    \item  Let $v_1, v_2, \dots $ be the nodes of the History-wdag. When running the MT algorithm with resampling table $R$, each event $L(v_k)$ was true at time $k$, and each variable $X(i)$ at time $k$ takes value $R(i,j)$ where $j$ is the number of previous events involving variable $i$. This is precisely configuration $X_{v_k,R}$. So each $L(v_k)$ is true on configuration $X_{v_k,R}$, and the History-wdag is compatible with $R$.
    \item For any $v \in G$, the configuration $X_{v,R}$ has the distribution $\Omega$. So each event $L(v)$ holds on $X_{v,R}$ with probability  precisely $p(L(v))$. Also, determining if each $L(v)$ holds on $X_{v,R}$ involves disjoint entries from $R$. So the events are independent and the overall probability that $G$ is compatible with $R$ is $\prod_{v \in G} p(L(v)) = w(G)$. 
    \item Suppose the MT algorithm runs for $t$ iteration, and let $v_1, \dots, v_t$ be the nodes of the History-wdag $H$. Then $H(v_1), \dots, H(v_t)$  are distinct single-sink wdags which are all compatible with $R$, i.e. $| \mathfrak S\cap \mathfrak R | \geq t$. Also, if the MT algorithm uses entry $R(i,j)$ at some time $s $, then there must be at least $j$ vertices $w$ in $H(v_s)$ with $i \in \var(L(w))$ (namely, the prior resampling of variable $i$). Hence $j \leq |H(v_s)| \leq \maxsize (\mathfrak S\cap \mathfrak R)$. \qedhere
\end{enumerate}
\end{proof}

\begin{proof}[Proof of Theorem~\ref{mtdist-thm1}]
  Let $B_1, \dots, B_t$ be an MT-execution with History-wdag $H$ on nodes $v_1, \dots, v_t$, and for $i = 0, \dots, t$ let $X_i$ be the configuration after resampling $B_1, \dots, B_i$. So $X_0$ is the initial row of $R$ and $X_t = X$ is the final configuration and $E$ holds on $X_t$. Let $s$ be the minimal value where $E$ holds on $X_s$. 
    
Define $U = \{v_i  \mid i \leq s, L(v_i) \in \Gamma(E)\}$ and $H' = H(U)$. Since $H'$ is a prefix of $H$, it is compatible with $R$. Also, $\sink(H') \subseteq \Gamma(E)$. The configuration $X_{\text{root},R}$ for $H'$ agrees with $X_s$ on all variables in $\var(E)$, so $E$ holds on it.  Finally, if $\Pr_{\Omega}(L(v_i) \cap \neg E) = 0$ for some $v_i \in U$, then $L(v_i)$ was true at time $i$ and so $E$ must be true after the resampling at time $i-1$, contradicting minimality of $s$.  So $L(v) \in \mathcal B^E$ for all $v \in H$ and so $H' \in \mathfrak R'_E$.
\end{proof}

\begin{proof}[Proof of Proposition~\ref{gb2}]
  \begin{enumerate}
  \item The bound $W_{\epsilon} \geq w_q(\mathfrak F)$ is obvious from the definition of $\mathfrak F$.
  We have $w_q(\mathfrak C_B) \geq 1$ and $w_q(\mathfrak C'_E) \geq 1$ since each of these includes the empty wdag which has weight one. Also, consider the event $B \in \mathcal B$ with $p(B) = \pmax$. We can form a wdag $G \in \mathfrak S_B$ by taking any wdag $G' \in \mathfrak S_B$ and adding a new sink node labeled $B$, so $w_q(\mathfrak S_B) \geq q(B) (1 + w_q( \mathfrak S_B))$, i.e. $\mu(B) \geq \frac{q(B)}{1 - q(B)}$ and hence $w(\mathfrak C_B) \geq \frac{1}{1 - q(B)} = \frac{1}{1 - \pmax^{1-\eps}} \geq \frac{1}{1 - \pmax}$. 
  \item Since $w_q(G) \leq \tau$ for any $G \in \mathfrak F^{\text{low}}_{\tau}$, we have:
   $$
   w( \mathfrak F^{\text{low}}_{\tau} ) = \sum_{G \in {\mathfrak F}^{\text{low}}_{\tau}} w_q(G)^{\frac{1}{1-\eps}} \leq \sum_{G \in {\mathfrak F}^{\text{low}}_{\tau}} w_q(G) \tau^{\frac{1}{1-\eps} - 1} \leq \tau^{\frac{1}{1-\eps} - 1} w_q(\mathfrak F) \leq \tau^{\eps} w_q(\mathfrak F) \leq \tau^{\eps} W_{\eps}.
   $$
\item We first count the wdags $G \in \mathfrak F_{\tau}$ with weight at least $\tau^2$:
$$
 \sum_{\substack{G \in \mathfrak F_{\tau}: w_q(G) \geq \tau^2}} 1 \leq \frac{1}{\tau^2} \sum_{\substack{G \in \mathfrak F_{\tau}}} w_q(G) \leq \frac{w_q(\mathfrak F)}{\tau^2} \leq \frac{W_{\eps}}{\tau^2}.
$$

We next count the single-sink wdags in $\mathfrak F_{\tau}$ which satisfy $w_q(G - v) \geq \tau$:
$$
\negthickspace \negthickspace \negthickspace \sum_{G: w_q(G - v) \geq \tau} \negthickspace \negthickspace \negthickspace 1 \leq \negthickspace  \sum_{ \substack{ G \in \mathfrak F_{\tau} \cap \mathfrak S}} \negthickspace  \frac{ w_q(G-v) }{\tau}  = \sum_{B \in \mathcal B} \sum_{ \substack{ G \in \mathfrak F_{\tau} \cap \mathfrak S_B }} \frac{ w_q(G) }{\tau q(B)} \leq \sum_{B \in \mathcal B} \frac{ w_q(\mathfrak S_B) }{\tau q(B)} = \sum_{B \in \mathcal B} \frac{w_q(\mathfrak C_B)}{\tau} \leq \frac{W_{\eps}}{\tau}.
$$

\item First consider any wdag $G \in \mathfrak F_{\tau}$ with  $w_q(G) \geq \tau^2$, and let $t = |G|$. We claim that  $t \leq \min\{ |\mathfrak F_{\tau}|, 4 W_{\eps} \log(1/\tau) \}$.  First, for each $v \in G$, the single-sink wdag $G(v)$ is in $\mathfrak F_{\tau}$, and all such $G(v)$ are distinct; thus $|\mathfrak F_{\tau}| \geq t$. Next, we have  $w_q(B) = w(B)^{1-\eps} \leq \pmax^{(1-\eps) t}$, and so $t \leq \frac{2 \log (1/\tau)}{(1-\eps) \log (1/\pmax)} \leq  \frac{2 \log (1/\tau)}{(1-\eps) (1 - \pmax)}$; by part (1) above we have $\frac{1}{1-\pmax} \leq W_{\eps}$.

Next, consider any wdag $G$ with a single sink node $v$ where $w_q(G - v) \geq \tau$. We have shown that $|G-v| \leq \min\{ |\mathfrak F_{\tau}|, 4 W_{\eps} \log(1/\tau) \}$, and so $|G| \leq \min\{ |\mathfrak F_{\tau}|, 4 W_{\eps} \log(1/\tau) \} + 1$. \qedhere
     \end{enumerate}
     \end{proof}
  
\section{Algorithm to enumerate wdags}
\label{gb2a-proof}

The algorithm here works  by successively merging smaller wdags. Formally, for wdags $G, H$, we form the new wdag $\text{Concat}(G,H)$  by first taking the union of $G$ and $H$; then, for any vertices $u \in G$ and $v \in H$ with $L(u) \sim L(v)$, we add an edge from $u$ to $v$. We thus use the following algorithm to enumerate $\mathfrak F_{\tau}$; here $T$ is a parameter we will set later.

\begin{algorithm}[H]
\centering
\begin{algorithmic}[1]
\State Initialize $\mathfrak A_0$ to contain $m$ graphs with a singleton node labeled by each $B \in \mathcal B$
\For{$i = 1, \dots, T$}
\State Initialize $\mathfrak A_i \leftarrow \mathfrak A_{i-1}$
\For{all pairs $G_1, G_2 \in \mathfrak A_{i-1}$}
\State Form $H = \concat(G_1,  G_2)$ 
\State Check if $H \in \mathfrak F_{\tau}$; if so, add it to $\mathfrak A_i$
\EndFor
\EndFor
\State Return $\mathfrak A_T$
\end{algorithmic}
\caption{Algorithm to enumerate $\mathfrak F_{\tau}$}
\label{enuma1}
\end{algorithm}

To analyze this procedure, we first show a few useful results:
\begin{proposition}
If $H \trianglelefteq G$, then $G = \concat(H,  G - H)$.
\end{proposition}
\begin{proof}
Clearly, the nodes of $G$ and $\concat(G,H)$ are both the union of the nodes of $H$ and $G -H$. Also, any edge within $H$ or within $G - H$ is the same within $G$ and $\concat(H, G-H)$. So, consider some nodes $u \in H, v \in G-H$ with $L(u) \sim L(v)$, where $\concat(H, G-H)$ has a directed edge $(u,v)$. This must be in $G$ as well, as if $G$ has a directed edge $(v,u)$,  since $u \in H$ and $H$ is a prefix,  then $v$ would also be placed into $H$.
\end{proof}

\begin{proposition}
\label{gaav1}
For a wdag $G \in \mathfrak F$ with $|G| \geq 2$, there is some $H \trianglelefteq G$ with $H \in \mathfrak F$ and $|G|/4 \leq |H| \leq |G|/2$.
\end{proposition}
\begin{proof}
Let $H$ be chosen to satisfy the properties that (i) $H \in \mathfrak F$ and (ii) $H \trianglelefteq G$ and (iii) $|H| \geq |G|/4$; among all such $H$ satisfying these three conditions, choose the one which has smallest size. Note that is well-defined since $H = G$ satisfies the conditions. Let $v_1, \dots, v_s$ be the sink nodes of $H$.

If $s = 1$, then $H - v_1 \in \mathfrak C_{L(v_1)} \subseteq \mathfrak C_B$ and also $H - v_1 \trianglelefteq G$. By minimality of $H$, we must then have $|H - v_1| < |G|/4$; since $|G| \geq 2$ this implies $|H| \leq |G|/2$.

If $s > 1$, then consider the wdags $H_1 = H(v_1), H_2 = H(v_2, \dots, v_s)$. Every node of $H$ is in $H_1$ or $H_2$ or both, so at least one of the wdags $H_i$ has $|H_i| \geq |H|/2$.  Each $H_i$ is in $\mathfrak F$  (its sink nodes are a subset of those of $H$), so by minimality of $H$ we have $|H|/2 \leq |H_i| < |G|/4$, i.e. $|H| < |G|/2$.
\end{proof}

\begin{theorem}
For $T = \Omega(\log \maxsize(\mathfrak F_{\tau}))$, we have $\mathfrak F_{\tau} = \mathfrak A_T$.
\end{theorem}
\begin{proof}
We claim that if any wdag $G \in \mathfrak F_{\tau}$ has size at most $(4/3)^{i}$ and has $w_q(G) \geq \tau^2$, then $G \in \mathfrak A_i$. We show it by induction on $i$; the base case $i = 0$ is clear since $\mathfrak A_0$ contains all the singleton wdags. For the induction step, consider $G \in \mathfrak F_{\tau}$ with $|G| \leq (4/3)^{i}$. If $|G|=1$ we have $G \in \mathfrak A_0 \subseteq \mathfrak A_{i}$. Otherwise, by Proposition~\ref{gaav1}, there is $H \trianglelefteq G$ with $|G|/4 \leq |H| \leq |G|/2$ and $H \in \mathfrak F$. We have $G - H \in \mathfrak F$, since its sink nodes are a subset of those of $G$. Since $w_q(G) \geq \tau^2$, both $H$ and $G - H$ are in $\mathfrak F_{\tau}$ and have weight at least $\tau^2$. Also, $H$ and $G - H$ have size at most $3 |G|/4 \leq (4/3)^{i-1}$. By induction hypothesis, both $H$ and $G-H$ are in $\mathfrak A_{i-1}$, so $G = \concat(H, G-H)$ is placed into $\mathfrak A_{i}$.

Thus all wdags $G \in \mathfrak F_{\tau}$ with $w_q(G) \geq \tau^2$ are contained in $\mathfrak A_{T'}$ for $T' = \Omega(\log \maxsize(\mathfrak F_{\tau}))$.   Now, consider any wdag $G \in \mathfrak F_{\tau}$ that is not already in $\mathfrak A_{T'}$. Necessarily $G$ has a single sink node $v$ and $w_q(G - v) \geq \tau$. Then $G = \concat(G-v, v)$ is added to $\mathfrak A_{T'+1}$.  So $T = T' + 1$ satisfies the claim.
\end{proof}

Putting together these results, we conclude with the following main theorem:
\begin{theorem}
\label{enum-alg}
Let $b = \maxsize(\mathfrak F_{\tau})$ and $\phi = | \mathfrak F_{\tau} | m'$. Algorithm~\ref{enuma1}  uses $\poly(\phi)$ processors and $\tilde O( \log b \log \phi ) \leq \tilde O(\log^2 \phi)$ time to enumerate $\mathfrak F_{\tau}$. In particular, there is a sequential algorithm to enumerate $\mathfrak F_{\tau}$ in $\poly( \phi)$ time.
\end{theorem}
\begin{proof}
 We can compute $\concat(G,H)$ from $G,H$ using $\tilde O( \log(|G|H|m'))$ time and $\poly(|G|, |H|, m')$ processors. Thus each individual iteration can be executed with $NC^1(b m \phi)$ complexity; note from Proposition~\ref{gb2}(4) that $b \leq O(\phi)$.  After $T = O(\log b)$ iterations, all of $\mathfrak F_{\tau}$ gets enumerated.
\end{proof}

\section{Proof of Theorem~\ref{fool-thm2}}
\label{appendix2}

As a starting point, we use the following result of \cite{harris2} for binary-valued probability spaces.
\begin{theorem}[\cite{harris2}]
  \label{fool-thm}
  Suppose the variables $X(1), \dots, X(n)$ are iid Bernoulli-$1/2$ and $F_1, \dots, F_k$ are automata on them, and let $\phi = \max\{ \size(F_i), n, k, 1/\eps \}$.
 Then there is a deterministic parallel algorithm to find a distribution $D$ of support size $|D| = \poly(\phi)$ which fools $F_1, \dots, F_k$ to error $\eps$. The algorithm has a complexity of $\poly(\phi)$ processors and $\tilde O(\log \phi \log n)$ time.
\end{theorem}

We now show Theorem~\ref{fool-thm2} by extending Theorem~\ref{fool-thm} to allow other possibilities for the alphabet $\Sigma$ and probability distribution $\Omega$. 

\begin{proof}[Proof of Theorem~\ref{fool-thm2}]

(1)  Let $\Omega'$ be the probability space obtained by quantizing the probability distribution on each variable $X(i)$ to multiple of $2^{-b}$ for $b = \Omega(\log(\sigma n/\eps))$. A straightforward coupling argument between $\Omega$ and $\Omega'$ shows that $\Omega'$ fools $F_1, \dots, F_k$ to error $\eps/2$.  We can encode the probability distribution $\Omega'$ by replacing each variable $X(i)$ with $b$ independent Bernoulli-$1/2$ variables $X'(i,1), \dots, X'(i,\ell)$. Also, we can simulate each automaton $F_i$ on the original variables $X$ by an automaton $F_i'$ on the expanded variables $X'$ by adding $2^b$ additional states. Thus, $\size(F'_i) \leq \size(F_i) + 2^b$, and the number of variables is increased to $n' = n b$.

  We apply Theorem~\ref{fool-thm} to generates a distribution $D$ over $\{0, 1 \}^{n'}$ fooling $F'_1, \dots, F'_k$ to error $\eps/2$. Here $D$ can also be viewed as a distribution over $\Sigma^n$, which fools $F_1, \dots, F_k$ to error $\eps$.

(2) We use the first result to construct a distribution $D$ fooling the automata to error $\rho' = \frac{\rho^2}{80 k}$.  
  Let us define $p_i = \Pr_{\Omega}(E_i)$, which  can be computed efficiently using Observation~\ref{simple-dynamic}.    Note that, since $D$ fools each automaton $F_i$, we have $\Pr_D(E_i) \leq  p_i + \rho'$ for each $i$. Let $U \subseteq \{1, \dots, k \}$ be the set of indices $i$ with $p_i \leq \nu := \frac{\rho}{40 k}$, and consider the function: 
  $$
  \Phi(X) =\sum_{i \in U} E_i(X)  + \frac{\sum_{i \notin U} s_i E_i(X)}{10 \sum_{i \notin U} s_i p_i}.
  $$
  
  We can calculate the expectation as:
  $$
  \mathbf E_D [\Phi(X)] \leq \sum_{i \in U} (p_i + \rho') + \frac{\sum_{i \notin U} s_i (p_i + \rho')}{10 \sum_{i \notin U} s_i p_i} \leq k (\nu  + \rho') + \frac{1}{10} + \frac{\rho'}{10 \nu} \leq \frac{1+\rho}{10}.
  $$
  
  Since $D$ has polynomial size, we can search the entire space to find a configuration $X$ with $\Phi(X) \leq (1+\rho)/10$. Since $\Phi(X) < 1$, this implies that $E_i(X) = 0$ for all $i \in U$. We then get
  \[
  \sum_{i \in [n]} s_i E_i(X) = \sum_{i \notin U} s_i E_i(X) = \Phi(X) \cdot (10 \sum_{i \notin U} s_i p_i)  \leq (1 + \rho)/10 \cdot (10 \sum_{i \in [n]} s_i p_i) = \sum_{i \in [n]} s_i p_i. \qedhere
  \]
\end{proof}

    \section{Simulating the MT algorithm}
    \label{sim-appendix}

    We will show the following result for the simulation of the MT algorithm:
    
\begin{lemma}
  \label{lem-post}
  Given a resampling table $R$ and an explicit listing of $\mathfrak S\cap \mathfrak R(R)$, there is an algorithm with $NC^2( \phi)$ complexity to find a good configuration $Y$ which is the output of a MT-execution on $R$, where $\phi = \max\{m,n, | \mathfrak S \cap \mathfrak R(R) | \}$.
\end{lemma}

    The algorithm  depends on the following definition taken from \cite{mt3}:
    \begin{definition}[Consistent wdags]
For a wdag $G$ and $i \in [n]$, we define $G[i]$ to the induced subgraph on all vertices $w \in G$ with $L(w) \in \mathcal B_i$. We say a pair of wdags $G, H$ are \emph{consistent} if for each $i \in [n]$, either $G[i]$ is an initial segment of $H[i]$ or $H[i]$ is an initial segment of $G[i]$.
    \end{definition}

Based on this definition, we use the following algorithm: we form the graph $H$ whose nodes are the elements of $\mathfrak S\cap \mathfrak R$, and where there is an edge on nodes corresponding to wdags $G, G'$ if $G, G'$ are inconsistent.  Using the standard algorithm of \cite{luby-mis}, we compute an MIS $\mathfrak I$ of $H$. We then output the configuration $Y_{\mathfrak I, R}$ defined by $Y_{\mathfrak I, R}(i) = R(i,\max_{G \in \mathfrak I} |G[i]|)$ for all $i$.

The algorithm clearly has the claimed complexity. Its correctness is implied by the following general result on maximal consistent wdag sets. (Note here that $\mathfrak I$ is precisely such a set).
     \begin{proposition}
    \label{lem-post1}
    Suppose that $R$ is a resampling table and $\mathfrak A \subseteq \mathfrak S \cap \mathfrak R(R)$ is a set of wdags  with the property that (i) every pair of wdags in $\mathfrak A$ are consistent, and (ii) for any $G' \in (\mathfrak S \cap \mathfrak R(R)) - \mathfrak A$ there is some $G \in \mathfrak A$ which is inconsistent with $G'$. Then  $Y_{\mathfrak A, R}$ is a good configuration which is the output of an  MT-execution on $R$.
  \end{proposition}
  \begin{proof}
We prove this by induction on $s_{\mathfrak A} := \sum_{G \in \mathfrak A} |G|$.  For the base case $s_{\mathfrak A} = 0$, we have $\mathfrak A = \emptyset$. Then $Y_{\mathfrak A, R}$ is trivially the output of an MT-execution (with no resamplings) on $R$. Also, it is a good configuration, as if some $B$ is true, then the wdag $G'$ consisting of a singleton node labeled $B$ is compatible with $R$, but there is clearly no $G \in \mathfrak A$ inconsistent with $G'$. 

For the induction step, let $G$ be an arbitrary element of $\mathfrak A$, let $v$ be an arbitrary source node of $G$, and let $B = L(v)$. Define $\mathfrak A_0$ to be the set of wdags $H \in \mathfrak A$ such that $L(u) \nsim B$ for all $u \in H$, and $\mathfrak A_1$ to be the set of wdags $H \in \mathfrak A$ with a (necessarily unique) source node labeled $B$.

We claim that $\mathfrak A = \mathfrak A_0 \sqcup \mathfrak A_1$. For, consider $H \in \mathfrak A - \mathfrak A_0$, and suppose that $u$ is the earliest node of $H$ such that $L(u) \sim B$. Consider some variable $i \in \var(B) \cap \var(B')$. The graph $G[i]$ starts with a node labeled $B$, while $H[i]$ starts with a node labeled $B'$. Since $G, H$ are consistent, it must be that $B' = B$. Furthermore, $H$ cannot contain a directed edge $(u', u)$, as then $L(u') \sim L(u) = B$, contradicting our choice of $u$. Thus $H$ has a source node $u$ labeled $B$ and so $u \in \mathfrak A_1$.

Now form the set $\mathfrak A'$ from $\mathfrak A$ by deleting  the source node labeled $B$ from each $H \in \mathfrak A_1$, and define a new resampling table $R'$ by ``shifting out'' $B$, i.e. set $R'(i,j) = R(i,j-1)$ for $i \in \var(B)$, and $R'(i,j) = R(i,j)$ otherwise. We claim that $\mathfrak A'$ satisfies the induction property with respect to resampling table $R'$. 

First, it is immediate from the definitions that every pair of wdags in $\mathfrak A'$ are consistent and that every $H' \in \mathfrak A'$ is compatible with $R'$. Furthermore, suppose that some $H' \in \mathfrak S \cap \mathfrak R(R')$ is consistent with all $G' \in \mathfrak A'$. Form $H$ from $H'$ by adding a new source node labeled $B$, with an edge to any $w \in H'$ such that $L(w) \sim B$. Since $H'$ is compatible with $R'$, this $H$ is compatible with $R$. Also, since $H'$ is consistent with every wdag in $\mathfrak A'$,  $H$ is consistent with every wdag in $\mathfrak A$. By our hypothesis on $\mathfrak A$, we must have $H \in \mathfrak A$. This implies that $H' \in \mathfrak A'$, since $H'$ is obtained from $H$ by deleting the source node labeled $B$.

 Finally, observe that $Y_{\mathfrak A', R'} = Y_{\mathfrak A, R}$.  At least one wdag in $\mathfrak A'$ has its size reduced (by removing a source node) compared to $\mathfrak A$, so $s_{\mathfrak A'} < s_{\mathfrak A}$. By the induction hypothesis, $Y_{\mathfrak A', R'}$ is a good configuration and is the output of an MT-execution on $R'$. Since $B$ holds on the initial row of $R$, this implies that $Y_{\mathfrak A, R}$ is the output of an MT-execution on $R$.
  \end{proof}

  \end{document}